\documentclass{lmcs} 
\pdfoutput=1

\usepackage{lastpage}
\lmcsdoi{18}{1}{24}
\lmcsheading{}{\pageref{LastPage}}{}{}%
{Jun.~09,~2020}{Feb.~01,~2022}{}

\usepackage[utf8]{inputenc}


\usepackage[english]{babel}

\usepackage{colortbl}
\usepackage{booktabs}
\usepackage{algorithm}
\usepackage[noend]{algpseudocode}
\algrenewcommand{\algorithmiccomment}[1]{\hfill$\triangleleft$ {\footnotesize \textsl{#1}}}
\algrenewcommand\algorithmicindent{2.0em}
\makeatletter
\renewcommand{\ALG@beginalgorithmic}{\small}
\makeatother

\usepackage{hhline}
\usepackage{multirow}

\usepackage{xspace}

\usepackage{amssymb}

\usepackage{tikz}
\usetikzlibrary{decorations.pathreplacing,arrows,decorations.pathmorphing,automata}
\usetikzlibrary{fit, patterns, calc,shapes,arrows,positioning}

\newcommand{\obsD}{\cO bs_D}
\newcommand{\reqD}{\cR eq_D}
\newcommand{\rank}{rank}
\newcommand{\row}{row}

\newcommand{\hsBE}{\textup{BE}}
\newcommand{\hsD}{\textup{D}}
\newcommand{\hsB}{\textup{B}}
\newcommand{\hsE}{\textup{E}}

\newcommand{\hsAAbar}{\textup{A$\overline{\textup{A}}$}}
\newcommand{\hsDhom}{\textup{D}$|_{\cH om}$}
\newcommand{\hsDhomStrict}{\textup{D}$|_{{\cH om},\ssubint}$}

\newcommand{\bD}{[D]}

\newcommand{\genDphiS}{\begin{tikzpicture}
\node(A)[inner sep=0pt]{$\scriptstyle{\Dphi\ssubint}$};
\node[inner sep=0pt](B) at (0.7,0){};
\node[inner sep=0pt](C) at (-0.6,0){};
\draw[double](C) -- (A);
\draw[double,->](A) --(B);
\end{tikzpicture}}

\newcommand{\RelMIN}{\Rightarrow_{\varphi}^{min}}
\newcommand{\Atoms}{\cA_{\varphi}}
\newcommand{\AP}{\mathpzc{AP}}

\newcommand{\Rows}{\cR ows_{\varphi}}

\newcommand{\RowsMIN}{\cR ows_{\varphi}^{min}}

\newcommand{\bbD}{\mathbb{S}}
\newcommand{\bbI}{\mathbb{I}}
\newcommand{\bbP}{\mathbb{P}}
\newcommand{\LTL}{\text{\sffamily LTL}}
\newcommand{\CTL}{\text{\sffamily CTL}}
\newcommand{\CTLstar}{\text{\sffamily CTL*}}

\newcommand{\cA}{\mathcal{A}}
\newcommand{\cG}{\mathcal{G}}
\newcommand{\cH}{\mathcal{H}}

\newcommand{\cL}{\mathcal{L}}
\newcommand{\cO}{\mathcal{O}}
\newcommand{\cR}{\mathcal{R}}

\newcommand{\hA}{\hat{A}}
\newcommand{\hn}{\hat{n}}
\newcommand{\hm}{\hat{m}}

\newcommand{\om}{\overline{m}}

\newcommand \details[1]{}

\newcommand{\Psp}{\textbf{PSPACE}}
\newcommand{\EXPSPACE}{\textbf{EXPSPACE}}
\newcommand{\NPsp}{\textbf{NPSPACE}}
\newcommand{\NLOGSP}{\textbf{NLOGSPACE}}
\newcommand{\co}{\textbf{co-}}

\newcommand{\Ku}{\ensuremath{\mathpzc{K}}}
\newcommand{\States}{W}
\newcommand{\Edges}{E}

\newcommand{\sinit}{s_0}
\newcommand{\SUCC}{succ_\varphi}
\newcommand{\SUCCStrict}{succ_{\varphi,\ssubint}}
\newcommand{\SuccMIN}{succ_\varphi^{min}}

\newcommand{\Instance}{\mathcal{I}}
\newcommand{\tupleof}[1]{(#1)}
\newcommand{\Init}{\textit{init}}
\newcommand{\Final}{\textit{final}}
\newcommand{\Left}{\textit{left}}
\newcommand{\Right}{\textit{right}}
\newcommand{\Down}{\textit{down}}
\newcommand{\Up}{\textit{up}} 

\DeclareMathOperator{\CL}{CL}
\DeclareMathOperator{\REQ}{REQ}

\DeclareMathAlphabet{\mathpzc}{OT1}{pzc}{m}{it}

\begin{document}

\title[SAT and MC for D under homogeneity]{Satisfiability and Model Checking for the Logic of Sub-Intervals
under the Homogeneity Assumption\rsuper*}
\titlecomment{{\lsuper*}This paper is an extended and revised version of \cite{bmmpsICALP}.}

\author[L.~Bozzelli]{Laura Bozzelli\rsuper{a}}	

\author[A.~Molinari]{Alberto Molinari\rsuper{b}}	

\author[A.~Montanari]{Angelo Montanari\rsuper{b}}	

\author[A.~Peron]{Adriano Peron\rsuper{a}}	

\author[P.~Sala]{Pietro Sala\rsuper{c}}	

\address{University of Napoli ``Federico II'', Napoli, Italy}	
\email{\texttt{lr.bozzelli@gmail.com}, \texttt{adrperon@unina.it}}  

\address{University of Udine, Udine, Italy}	
\email{\texttt{molinari.alberto@gmail.com}, \texttt{angelo.montanari@uniud.it}}  

\address{University of Verona, Verona, Italy}	
\email{\texttt{pietro.sala@univr.it}}  

\begin{abstract}
  The expressive power of interval temporal logics (ITLs) makes them
one of the most natural choices in a number of application domains, ranging from the specification and verification of complex reactive systems to automated planning.
However, for a long time, because of their high computational complexity, they were considered not suitable for practical purposes.
The recent discovery of several computationally well-behaved ITLs has finally changed the scenario.

In this paper, we investigate the finite satisfiability and model checking problems for the ITL 
\hsD, that has a single modality for the sub-interval relation, under the homogeneity assumption (that constrains a proposition letter to hold over an interval if and  only if it holds over all its points).
We first prove that the satisfiability problem for \hsD, over finite linear orders, is $\Psp$-complete, and then
we show that the same holds for its model checking problem, over finite Kripke structures. 
In such a way, we enrich the set of tractable interval temporal logics with a new meaningful representative.
\end{abstract}

\maketitle
\section{Introduction}

For a long time, interval temporal logics (ITLs) were considered an attractive, but impractical, alternative to standard point-based ones. On the one hand, as stated by Kamp and Reyle~\cite{from_discourse_to_logic}, ``truth, as it pertains to language in the way we use it, relates sentences not to instants but to temporal intervals'', and thus ITLs are a natural choice as a specification/representation language; on the other hand, the high undecidability of the satisfiability problem for the most well-known ITLs, such as Halpern and Shoham's modal logic of time intervals (HS for short)~\cite{interval_modal_logic} and Venema's CDT~\cite{chopping_intervals}, prevented an extensive use of them (in fact, some quite restricted variants of them have been applied in formal verification and artificial intelligence over the years).

The present work fits into the context of the logic HS, which features one modality for each of the  thirteen Allen's relations, apart from equality. In Table~\ref{allen}, we depict six Allen's relations, each one corresponding to a different ordering relation between a pair of intervals,
together with the corresponding HS (existential) modalities. The other seven relations are the inverses
of the six depicted ones, plus the equality relation.

\begin{table}[tp]
	\renewcommand{\arraystretch}{1.1}
	\centering
	\caption{Allen's relations and the corresponding HS modalities.}\label{allen}
	\medskip
	\begin{tabular}{cclc}
		\toprule
		\rule[-1ex]{0pt}{3.5ex} Allen relation & HS & Definition w.r.t.\ interval structures &  Example\\
		\midrule
		
		&   &   & \multirow{7}{*}{\begin{tikzpicture}[scale=1.03]
\draw[draw=none, use as bounding box](-0.4,0.2) rectangle (3.4,-3.1);

\coordinate [label=left:\textcolor{red}{$x$}] (A0) at (0,0);
\coordinate [label=right:\textcolor{red}{$y$}] (B0) at (1.5,0);
\draw[red] (A0) -- (B0);
\fill [red] (A0) circle (2pt);
\fill [red] (B0) circle (2pt);

\coordinate [label=left:$v$] (A) at (1.5,-0.5);
\coordinate [label=right:$z$] (B) at (2.5,-0.5);
\draw[black] (A) -- (B);
\fill [black] (A) circle (2pt);
\fill [black] (B) circle (2pt);

\coordinate [label=left:$v$] (A) at (2,-1);
\coordinate [label=right:$z$] (B) at (3,-1);
\draw[black] (A) -- (B);
\fill [black] (A) circle (2pt);
\fill [black] (B) circle (2pt);

\coordinate [label=left:$v$] (A) at (0,-1.5);
\coordinate [label=right:$z$] (B) at (1,-1.5);
\draw[black] (A) -- (B);
\fill [black] (A) circle (2pt);
\fill [black] (B) circle (2pt);

\coordinate [label=left:$v$] (A) at (0.5,-2);
\coordinate [label=right:$z$] (B) at (1.5,-2);
\draw[black] (A) -- (B);
\fill [black] (A) circle (2pt);
\fill [black] (B) circle (2pt);

\coordinate [label=left:\textcolor{red}{$v$}] (A) at (0.5,-2.5);
\coordinate [label=right:\textcolor{red}{$z$}] (B) at (1,-2.5);
\draw[red] (A) -- (B);
\fill [red] (A) circle (2pt);
\fill [red] (B) circle (2pt);

\coordinate [label=left:$v$] (A) at (1.3,-3);
\coordinate [label=right:$z$] (B) at (2.3,-3);
\draw[black] (A) -- (B);
\fill [black] (A) circle (2pt);
\fill [black] (B) circle (2pt);

\coordinate (A1) at (0,-3);
\coordinate (B1) at (1.5,-3);
\draw[dotted] (A0) -- (A1);
\draw[dotted] (B0) -- (B1);
\end{tikzpicture}}\\
		
		\textsc{meets} & $\langle A \rangle$ & $[x,y]\mathpzc{R}_A[v,z]\iff y=v$ &\\
		
		\textsc{before} & $\langle L\rangle$ & $[x,y]\mathpzc{R}_L[v,z]\iff y<v$ &\\
		
		\textsc{started-by} & $\langle B \rangle$ & $[x,y]\mathpzc{R}_B[v,z]\iff x=v\wedge z<y$ &\\
		
		\textsc{finished-by} & $\langle E \rangle$ & $[x,y]\mathpzc{R}_E[v,z]\iff y=z\wedge x<v$ &\\
		
		\textcolor{red}{\textsc{contains}} & \textcolor{red}{$\langle D \rangle$} & $[x,y]\mathpzc{R}_D[v,z]\iff [v,z] \subset [x,y]$ &\\
		
		\textsc{overlaps} & $\langle O \rangle$ & $[x,y]\mathpzc{R}_O[v,z]\iff x<v<y<z$ &\\
		
		\bottomrule
	\end{tabular}
\end{table}

The recent discovery of a significant number of expressive and computationally well-behaved fragments of full HS  changed the landscape of research in ITLs (~\cite{DBLP:journals/eatcs/MonicaGMS11,DBLP:conf/time/Montanari16}. Meaningful examples are the logic \hsAAbar\ of temporal neighborhood~\cite{DBLP:journals/apal/BresolinGMS09} (the HS fragment featuring modalities for the \emph{meets} relation and its inverse \emph{met by}) and the logic \hsD\ of (temporal) sub-intervals~\cite{DBLP:journals/logcom/BresolinGMS10} (the HS fragment featuring only one modality for the \emph{contains} relation). In this paper, we focus on the latter.

To begin with, we observe that  \hsD{} is a (proper) fragment of the logic \hsBE\ with modalities for the Allen's relations \emph{started-by} and \emph{finished-by} (see Table~\ref{allen}) since any sub-interval is just an initial sub-interval of an ending one, or, equivalently, an ending sub-interval of an initial one.
From a computational point of view, \hsD\ is a real character.
Indeed, it has been shown that its satisfiability problem is $\Psp$-complete over the class of dense linear orders~\cite{DBLP:journals/logcom/BresolinGMS10,DBLP:conf/aiml/Shapirovsky04} (the same problem is undecidable for \hsBE~\cite{DBLP:conf/asian/Lodaya00}), while it becomes undecidable when the logic is interpreted over the classes of finite and discrete linear orders~\cite{DBLP:journals/fuin/MarcinkowskiM14}. The situation is still unknown over the class of all linear orders. As for its expressiveness, unlike \hsAAbar---which is expressively complete with respect to the two-variable fragment of first-order logic for binary relational structures over various linearly-ordered domains~\cite{DBLP:journals/apal/BresolinGMS09,DBLP:journals/jsyml/Otto01}---three variables are needed to encode \hsD\ in first-order logic (the two-variable property is a sufficient condition for decidability, but it is not a necessary one).

In this paper, we show that the decidability of the satisfiability problem for \hsD\ over the class of finite linear orders can be recovered by assuming homogeneity (the \emph{homogeneity assumption} constrains a proposition letter to hold over an interval if and only if it holds over all its constituent points).
First, by exploiting a suitable contraction method, we prove that, under the homogeneity assumption, the problem belongs to $\Psp$. Then, we show that the proposed  algorithm  for satisfiability checking can be transformed into a $\Psp$ model-checking procedure for \hsD\ formulas over finite Kripke structures (under homogeneity). Finally, $\Psp$-hardness of both satisfiability and model checking is proved via a reduction from a domino-tiling problem for grid with rows of linear length. 

The result of this paper about $\Psp$-completeness of the model checking problem for \hsD\  is particularly interesting when compared to the model checking problem for \hsBE, which includes \hsD\ as a proper fragment and, at first sight, may seem to be quite close to \hsD.
The exact complexity of the model checking problem for \hsBE, over finite Kripke structures, under the homogeneity assumption, is still unknown and it is a difficult open question whether it can be solved elementarily. At the present time, only a nonelementary model checking procedure is known for it~\cite{DBLP:journals/acta/MolinariMMPP16}, and the closest proved lower bound is $\EXPSPACE$ \cite{TCSBMMPS19}. The complete picture for the complexity of the model checking (MC for short) problem for HS and its fragments (under homogeneity) is reported in Figure~\ref{fGr}. Being \hsD\ the most significant fragment of \hsBE, the results
given in this paper allow us to gain a better insight into such an open question. The exact complexity of the satisfiability problem for \hsBE\ over finite linear orders (under homogeneity) is an open issue as well: the same upper/lower bounds can be shown to hold by reductions to/from the MC problem.

\paragraph{Related work.}
Temporal logics of sub-intervals have been already investigated in the literature. In particular, they come into play in the study of temporal prepositions in natural language~\cite{DBLP:journals/ai/Pratt-Hartmann05}.
The connections between the temporal logic of (strict) sub-intervals and the logic of Minkowski space-time have been explored by Shapirovsky and Shehtman~\cite{DBLP:journals/logcom/ShapirovskyS05}, while
the temporal logic of reflexive sub-intervals has been studied for the first time by van Benthem, who proved that, when interpreted over dense linear orderings, it is equivalent to the standard modal logic S4~\cite{DBLP:journals/jphil/Benthem91}.
Satisfiability for the temporal logic of reflexive sub-intervals over the class of finite linear orders is decidable and $\Psp$-complete~\cite{DBLP:conf/time/MontanariPS10}.

The satisfiability problem for HS has been extensively studied in the literature. The problem in its full generality, that is, without the homogeneity assumption, has been proved to be highly undecidable over all relevant (classes of) linear orders \cite{interval_modal_logic}. The same holds for most of its fragments~\cite{DBLP:journals/amai/BresolinMGMS14,DBLP:conf/asian/Lodaya00,DBLP:journals/fuin/MarcinkowskiM14}.

The MC problem for HS and its fragments over finite Kripke structures (under the homogeneity assumption)
has been systematically explored in a series of papers~\cite{LM13,LM14,LM16,DBLP:journals/acta/MolinariMMPP16,ICMMP18,ICBMMPS18,TCSBMMPS19}.
In Figure~\ref{fGr}, we add the contribution of this paper to the general picture of known complexity results
making it clear that it enriches the set of ``tractable'' fragments with a meaningful new member.
An expressive comparison of the MC problem for HS and 
point-based temporal logics \LTL, \CTL, and \CTLstar\ can be found in \cite{TOCLMMPS19}.

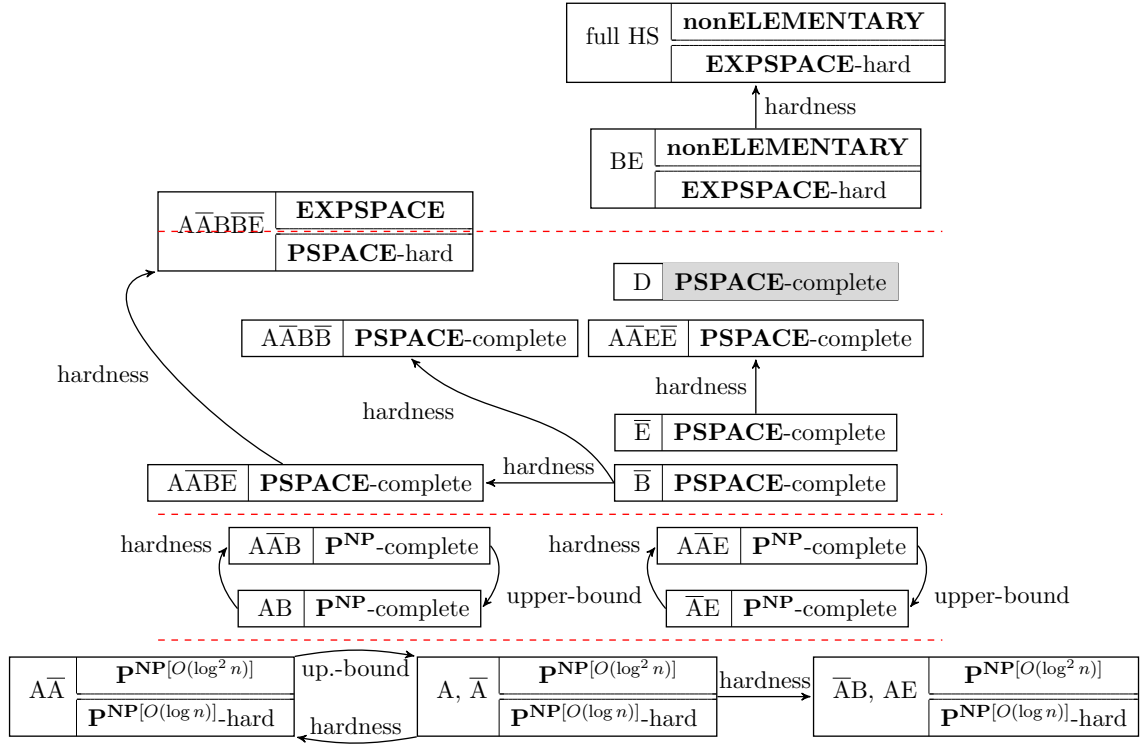
\begin{figure}[t]
	\centering
	\newcommand{\A}{\mathrm{A}}
\newcommand{\Abar}{\mathrm{\overline{A}}}
\newcommand{\AAbar}{\mathrm{A\overline{A}}}
\newcommand{\AAbarBbarEbar}{\mathrm{A\overline{A}\overline{B}\overline{E}}}
\newcommand{\ABbar}{\mathrm{A\overline{B}}}
\newcommand{\Bbar}{\mathrm{\overline{B}}}
\newcommand{\Ebar}{\mathrm{\overline{E}}}
\newcommand{\B}{\mathrm{B}}
\newcommand{\E}{\mathrm{E}}
\newcommand{\BE}{\mathrm{B}\mathrm{E}}
\newcommand{\AAbarB}{\mathrm{A\overline{A}B}}
\newcommand{\AAbarE}{\mathrm{A\overline{A}E}}
\newcommand{\HSexi}{\mathrm{\exists A\overline{A}BE}}
\newcommand{\HSforall}{\mathrm{\forall A\overline{A}BE}}
\newcommand{\ABBbar}{\mathrm{AB\overline{B}}}
\newcommand{\AAbarBBbarEbar}{\mathrm{A\overline{A}B\overline{B}\overline{E}}}
\newcommand{\AAbarEBbarEbar}{\mathrm{A\overline{A}E\overline{B}\overline{E}}}
\newcommand{\AAbarBBbar}{\mathrm{A\overline{A}B\overline{B}}}
\newcommand{\AAbarEEbar}{\mathrm{A\overline{A}E\overline{E}}}
\newcommand{\AAbarBE}{\mathrm{A\overline{A}BE}}
\newcommand{\HSprop}{\mathrm{Prop}}

\newcommand{\PTIME}{\mathbf{P}}
\newcommand{\NP}{\mathbf{NP}}
\newcommand{\LOGSPACE}{\mathbf{LOGSPACE}}
\newcommand{\NEXPTIME}{\mathbf{NEXPTIME}}
\newcommand{\PH}{\mathbf{PH}}
\newcommand{\coNP}{\mathbf{co-NP}}
\newcommand{\Th}{\PTIME^{\NP[O(\log n)]}}
\newcommand{\Thsq}{\PTIME^{\NP[O(\log^2 n)]}}

\newcommand{\cellThree}[3]{
\begin{tabular}{c|c}
\rule[-1ex]{0pt}{3.5ex}
\multirow{2}{*}{#1} & #2 \\ 
\hhline{~=}\rule[-1ex]{0pt}{3.5ex}
 & #3 
\end{tabular}}

\newcommand{\cellTwo}[2]{\begin{tabular}{c|c}
\rule[-1ex]{0pt}{3.5ex}
#1 & #2 \\
\end{tabular}}

\newcommand{\cellOne}[1]{\begin{tabular}{c}
		\rule[-1ex]{0pt}{3.5ex}
		#1
	\end{tabular}}

\resizebox{\linewidth}{!}{
\begin{tikzpicture}[->,>=stealth',shorten >=1pt,auto,semithick,main node/.style={rectangle,draw, inner sep=0pt}]  

\tikzstyle{gray node}=[fill=gray!30]

    \node [main node](0) at (-4,0) {\cellTwo{$\AAbarBbarEbar$}{$\Psp$-complete}};
    \node [main node](1) at (3,0)  {\cellTwo{$\Bbar$}{$\Psp$-complete}};
    \node [main node](21) at (3,0.8)  {\cellTwo{$\Ebar$}{$\Psp$-complete}};
    \node [main node](31) at (3,2.3)  {\cellTwo{$\AAbarEEbar$}{$\Psp$-complete}};
    \node [main node](31a) at (3,3.2)  {\cellTwo{$\mathrm{D}$}{\cellcolor{gray!30}{$\Psp$-complete}}};
    \node [main node](32) at (-2.5,2.3)  {\cellTwo{$\AAbarBBbar$}{$\Psp$-complete}};
    \node [main node](2) at (-6.6,-3.4) {\cellThree{$\AAbar$}{$\Thsq$}{$\Th$-hard }};
    \node [main node](22) at (0,-3.4) {\cellThree{$\A$, $\Abar$}{$\Thsq$}{$\Th$-hard }};
    \node [main node](92) at (6.5,-3.4) {\cellThree{$\Abar\B$, $\A\E$}{$\Thsq$}{$\Th$-hard}};
    
    \node [main node](72) at (-3.3,-1.0)  {\cellTwo{$\AAbarB$}{$\PTIME^{\NP}$-complete}};
    \node [main node](73) at (3.5,-1.0)  {\cellTwo{$\AAbarE$}{$\PTIME^{\NP}$-complete}};
    \node [main node](74) at (-3.3,-2)  {\cellTwo{$\A\B$}{$\PTIME^{\NP}$-complete}};
    \node [main node](75) at (3.5,-2)  {\cellTwo{$\Abar\E$}{$\PTIME^{\NP}$-complete}};
    
    
    \node [main node](5) at (-4,4) {\cellThree{$\AAbarBBbarEbar$}{$\EXPSPACE$ }{$\Psp$-hard }};
    
    \node [main node](9) at (3,5) {\cellThree{$\BE$}{\textbf{nonELEMENTARY}}{$\EXPSPACE$-hard }};
    \node [main node](10) at (3,7) {\cellThree{full HS}{\textbf{nonELEMENTARY} }{$\EXPSPACE$-hard}};
   
    \path
    (1) edge [swap] node {hardness} (0) 
    (0) edge  [out=150,in=190] node {hardness} (5.south west)
    (2.north east) edge [swap,out=370,in=170] node {up.-bound} (22.north west)
    (22.south west) edge [swap,out=190,in=-10] node {hardness} (2.south east)
    (22.east) edge node {hardness} (92.west)
    (9) edge [swap] node {hardness} (10)
    (21.north) edge node {hardness} (31.south)
    (1.west) edge [out=120,in=-45, near end] node {hardness} (32.south)
    (74.west) edge [out=130, in=230, near end] node {hardness} (72.west)
    (72.east) edge [out=310, in=50] node {upper-bound} (74.east)
    (75.west) edge [out=130, in=230, near end] node {hardness} (73.west)
    (73.east) edge [out=310, in=50] node {upper-bound} (75.east)
    ;
    
    \draw [dashed,-,red] (-6.5,4) -- (6,4);
    \draw [dashed,-,red] (-6.5,-2.5) -- (6,-2.5);
    \draw [dashed,-,red] (-6.5,-0.5) -- (6,-0.5);
\end{tikzpicture}
}
	\vspace{-0.4cm}
	\caption{Complexity of the MC problem for HS and its fragments (under homogeneity).}
	\label{fGr}
\end{figure}

\paragraph{Organization of the paper.}
In Section~\ref{sec:preliminaries}, we provide some background knowledge. In particular, we introduce the logic \hsD, the notion of interval model, and an extremely useful spatial representation of interval models called compass structure.
In Section~\ref{sec:decidability}, we prove $\Psp$ membership of the satisfiability problem for \hsD\ over finite linear orders
(under homogeneity) making use of a contraction technique applied to compass structures, which relies on a suitable equivalence relation of finite index.
In Section~\ref{sec:mc}, we show that the MC problem for  \hsD\  (under homogeneity), over finite Kripke structures, is in $\Psp$ as well. The proposed MC algorithm is basically a satisfiability procedure driven by the computation traces of the system model.
In Section~\ref{sec:hardness}, we show $\Psp$-hardness of both problems via a polynomial-time  reduction from a domino-tiling problem for grids with rows of linear length. 
Conclusions summarize the work done and outline future research directions.

\newcommand{\D}{\langle D\rangle}
\newcommand{\TD}{\langle \overline D\rangle}
\newcommand{\BD}{[D]}
\newcommand{\bM}{\mathbf{M}}
\newcommand{\cV}{\mathcal{V}}

\newcommand{\DCirc}{\langle D_\circ\rangle}


\newcommand{\A}{\ang{A}}
\newcommand{\BA}{[A]}
\newcommand{\TA}{\ang{\overline{A}}}
\newcommand{\BTA}{[\overline{A}]}
\newcommand{\CA}{(A)}
\newcommand{\CTA}{(\overline{A})}
\newcommand{\Dr}{\Diamond_r}
\newcommand{\Br}{\Box_r}
\newcommand{\Dl}{\Diamond_l}
\newcommand{\Bl}{\Box_l}

\newcommand{\RB}{\mathrel{R_{\cB}}}
\newcommand{\fS}{f_{\bS}}
\newcommand{\fbS}{f^{b}_{\bS}}
\newcommand{\feS}{f^{e}_{\bS}}



\newcommand{\Sse}[1]{%
  \stackrel{\text{#1}}{\Longleftrightarrow}%
}

\newcommand{\Allora}[1]{%
  \stackrel{\text{#1}}{\Longrightarrow}%
}
\newcommand{\subint}{\sqsubset}
\newcommand{\subinteq}{\sqsubseteq}
\newcommand{\ssubint}{{\text{$\sqsubset$\llap{$\;\cdot\;$}}}}
\newcommand{\ssupint}{{\text{$\sqsupset$\llap{$\;\cdot\;$}}}}
\newcommand{\Dphi}{\mathrel{D_\varphi}}
\newcommand{\bbG}{\mathbb{G}}
\newcommand{\Dref}{\textup{D$_\subinteq$}\xspace}
\newcommand{\Dirr}{\textup{D$_\subint$}\xspace}
\newcommand{\Dstr}{\textup{D$_\ssubint$}\xspace}
\newcommand{\Dsim}{\textup{D$_\circ$}\xspace}

\section{The logic \hsD\ of the sub-interval relation}\label{sec:preliminaries}

To begin with, we introduce some preliminary notions.
Let $\mathbb{S} = \langle S, <\rangle$ be a  linear order over a nonempty set $S\neq \emptyset$, and $\leq$
be the reflexive closure of $<$. Given two elements $x,y\in S$, with $x\leq y$, we denote by $[x,y]$
the (closed) \emph{interval} over $S$ consisting of the set of elements $z\in S$ such that $x\leq z$ and $z\leq y$.
For $z\in S$, we write $z\in [x,y]$ to mean that $z$ is an element of the interval $[x,y]$.
We denote the set of all intervals over
$\mathbb{S}$ by $\mathbb{I(S)}$.
We consider two possible \emph{sub-interval relations}:
\begin{enumerate}
    %
    \item the \emph{proper} sub-interval relation (denoted as $\subint$), defined
    by $[x, y]\subint  [x', y']$ if $x' \leq x$, $y\leq y'$, and $[x, y] \neq [x', y']$ (corresponding to the proper subset relation over intervals), 
    \item and the \emph{strict}
    sub-interval relation (denoted as $\ssubint$), defined by
    $[x, y] \ssubint [x', y'] $ if and only if $x' < x$ and $y<y'$.
\end{enumerate}

The two modal logics  $\Dirr$  and $\Dstr$ feature the same
language, consisting of a finite set $\AP$ of proposition
letters/variables, the logical connectives $\neg$ and $\vee$, and the
modal operator $\D$. Formally, formulae are
defined by the grammar: \[\varphi ::= p \ |\ \neg\varphi\ |\ \varphi \vee \varphi\ |\ \D \varphi,\]
with $p\in\AP$.
The other connectives, as well as
the logical constants $\top$ (true) and $\bot$ (false), are defined as usual; moreover, the dual universal modality $\BD\varphi$ is defined as $\neg\D\neg\varphi$. The length of a formula $\varphi$, denoted as $|\varphi|$, is the number of sub-formulas of $\varphi$.

The semantics of~\Dirr  and \Dstr  only differ in the
interpretation of the $\D$ modality. For the sake of brevity, we use
$\circ\  \in \{\subint,\ssubint\}$ as a shorthand for
any of the two sub-interval relations. The semantics of a
sub-interval logic \Dsim is defined in terms of \emph{interval models}
$\bM=\langle\mathbb{I(S)},\cV\rangle$, where $\cV : \AP \mapsto 2^{\mathbb{I(S)}}$ is a \emph{valuation function}  assigning to
every proposition letter $p$ the set of intervals $\cV(p)$ over
which it holds. Given an interval model $\bM=\langle \mathbb{I(S)}, \cV\rangle$, an interval $[x,y]\in \mathbb{I(S)}$, and a formula $\psi$,  the \emph{satisfaction relation} $\bM,[x,y]\models\psi$, stating that $\psi$ holds over the interval $[x,y]$ of $\bM$, is inductively defined as~follows:
%
\begin{itemize}
  \item for every proposition letter  $p \in \AP$, $\bM, [x,y]
        \models p$ \,\emph{if}\,  $[x,y] \in \cV(p)$;

  \item $\bM, [x,y] \models \neg \psi$ \,\emph{if}\,  $\bM, [x,y]
        \not\models \psi$, that is, it is not true that $\bM, [x,y]
        \models \psi$;

  \item $\bM, [x,y] \models \psi_1 \vee \psi_2$ \,\emph{if}\, either  $\bM, [x,y] \models \psi_1$ or $\bM, [x,y] \models \psi_2$;

  \item $\bM, [x,y] \models \D \psi$ \,\emph{if}\,   there exists an interval $[x',y'] \in
        \bbI(\mathbb{S})$ such that $[x',y'] \circ [x,y]$ and~$\bM, [x',y'] \models \psi$.
\end{itemize}
A \Dsim-formula is \Dsim-\emph{satisfiable} if it holds over some interval of an interval model, and it is \Dsim-\emph{valid} if it holds over every interval of every interval model.

%
In this paper, we restrict our attention to the finite satisfiability problem, that is, satisfiability over the class of finite linear orders. The problem has been shown to be \emph{undecidable} for $\Dirr$ and $\Dstr$ \cite{DBLP:journals/fuin/MarcinkowskiM14}.  
In the following, we show that decidability  for $\Dirr$ and $\Dstr$ can be recovered by restricting to the class of \emph{homogeneous} interval models. We fully work out the case of $\Dirr$ (for the sake of simplicity, hereafter we will write \hsD\ for $\Dirr$), and then we briefly explain how to adapt the proofs~to~$\Dstr$.

\begin{defi}\label{def:homogeneous_models}
A model
$\bM=\langle\mathbb{I(S)},\cV\rangle$ is \emph{homogeneous} if
for every interval $[x,y]  \in \mathbb{I(S)}$ and every $p \in \AP$,
it holds that $[x,y] \in \cV(p)$ if and only if $[x',x'] \in \cV(p)$
for every $x'\in [x,y]$.
\end{defi}
In the following, we will refer to \hsD\ interpreted over homogeneous models as \hsDhom.
Moreover, we refer to \Dstr\ interpreted over homogeneous models as \hsDhomStrict.

\subsection{Closure, atoms, and temporal requests}\label{sec:catr}

We now introduce some basic definitions and notation which will be
extensively used in the following. Given a \hsD-formula (resp., \Dstr-formula) $\varphi$,
we define the \emph{closure of $\varphi$}, denoted by
$\CL(\varphi)$, as the set of all sub-formulas $\psi$ of $\varphi$ and of their
negations $\neg\psi$ (we identify $\neg\neg \psi$ with $\psi$ and $\neg\D\psi$ with $\BD\neg\psi$). 

\begin{defi}\label{def:d-atom}
A \emph{$\varphi$-atom $A$} is a subset of
$\CL(\varphi)$ such that:
\begin{itemize}
    \item  for every $\psi \in \CL(\varphi)$, $\psi \in A$ if and only if $\neg\psi \notin A$, and
    \item for every $\psi_1 \vee \psi_2 \in \CL(\varphi)$, $\psi_1 \vee \psi_2 \in A$ if and only if either $\psi_1 \in A$ or $\psi_2 \in A$.
\end{itemize}
\end{defi}

Intuitively, a $\varphi$-atom describes a maximal set of sub-formulas of $\varphi$ which can hold at an interval of an interval model.
In particular, the idea underlying atoms is to enforce a \lq\lq local\rq\rq{} (or Boolean) form of consistency among the formulas it contains, that is, a $\varphi$-atom $A$ is a \emph{maximal, locally consistent subset} of $\CL(\varphi)$. As an example, $\neg(\psi_1 \vee \psi_2) \in A$ iff $\neg\psi_1\in A$ and $\neg\psi_2\in A$. Note, however, that the definition does not set any constraint on sub-formulas of $\varphi$ of the form $\D\psi$, hence the word ``local''.
We denote the set of all $\varphi$-atoms by $\cA_\varphi$. Its cardinality is clearly bounded by $2^{|\varphi|}$ (by the first item of Definition~\ref{def:d-atom}). Atoms are connected by the following binary relation $\Dphi$ which, intuitively, represents the ``symbolic'' counterpart of the inverse $\sqsupset$ of relation $\sqsubset$  between pairs of intervals, if $\varphi$ is a \hsD-formula,
and the inverse $\ssupint$ of relation $\ssubint$, otherwise.

\begin{defi}\label{def:Dphi-relation}
Let $\Dphi$ be the binary relation over $\cA_\varphi$ defined as follows: for each pair of atoms $A, A' \in \cA_\varphi$, $A \Dphi A'$ holds if
for all formulas $\BD\psi \in A$,
both $\psi \in A'$ and $\BD\psi \in A'$.
\end{defi}

Note that, by the semantics of \hsD\ (resp., \Dstr), if $A$ and $A'$ are the sets of formulas in $\CL(\varphi)$ that hold at  intervals $[x,y]$ and $[x',y']$, respectively, of an interval model such that $[x,y] \sqsupset [x',y']$ (resp., $[x,y]\, \ssupint\, [x',y']$), then $A \Dphi A'$.

Let $A$ be a $\varphi$-atom. We denote by $\reqD(A)$ the set $\{\psi \in \CL(\varphi) : \D \psi \in A\}$ of \lq\lq \emph{temporal requests}\rq\rq{} of $A$. In particular, if $\D\psi \in \CL(\varphi)$ and $ \psi \notin \reqD(A)$, then $\BD \neg \psi \in A$ (by definition of $\varphi$-atom). Moreover, we denote by $\REQ_{\varphi}$ the set of all  arguments of $\D$-formulas in $\CL(\varphi)$, namely, $\REQ_{\varphi}=\{ \psi: \D \psi\in \CL(\varphi)\}$. Finally, we denote by $\obsD(A)$ the set $\{\psi \in A : \psi \in \REQ_{\varphi}\}$ of \lq\lq \emph{observables}\rq\rq{} of $A$.

The next proposition, stating that, once proposition letters and temporal requests of $A$ have been fixed, $A$ gets unambiguously determined, can be easily proved by induction.
\begin{prop}\label{prop:unique}
    Let $\varphi$ be a \hsD-formula (resp., \Dstr-formula), $R \subseteq \REQ_{\varphi}$, and $P \subseteq \CL(\varphi)\cap \AP$. Then, there is a unique $\varphi$-atom $A$ that satisfies $\reqD(A)=R$ and $A\cap \AP = P$.
\end{prop}

In the rest of the paper, we will exploit the following characterization of the relation $\Dphi$.

\begin{prop}\label{prop:characterizeTemporalConsistency}
  Let $A$ and $A'$ be two $\varphi$-atoms. Then, $A \Dphi A'$ if and only if $\reqD(A') \subseteq \reqD(A)$ and $\obsD(A') \subseteq \reqD(A)$.
\end{prop}
\begin{proof}
We have that   $A \Dphi A'$ \emph{if and only if} for all
 $\BD\psi \in \CL(\varphi)$, $\BD\psi \in A$ entails that $\BD\psi \in A'$ and $\psi \in A'$ \emph{if and only if}
 (by definition of atom and the identification of $\neg\D\psi$ with $\BD\neg\psi$) for all
 $\D \psi \in \CL(\varphi)$, (i)  $\D\psi \in A'$ entails that $\D\psi \in A$, and (ii)  $\psi \in A'$ entails that $\D\psi \in A$
  \emph{if and only if} $\reqD(A') \subseteq \reqD(A)$ and $\obsD(A') \subseteq \reqD(A)$.
\end{proof}

\subsection{A spatial representation of interval models}\label{sec:compass}

We now provide a natural interpretation of \hsD\ (resp., \Dstr) over grid-like
structures, called \emph{compass structures}, by exploiting the existence
of a natural bijection between intervals $[x,y]$ and
points $(x,y)$, with $x \leq y$, of an $S\times S$ grid, where $\mathbb{S} = \langle S, <\rangle$ is a finite linear order. Such
an interpretation was originally proposed by Venema in~\cite{venema1990},
and it can also be given for HS and all its (other) fragments.

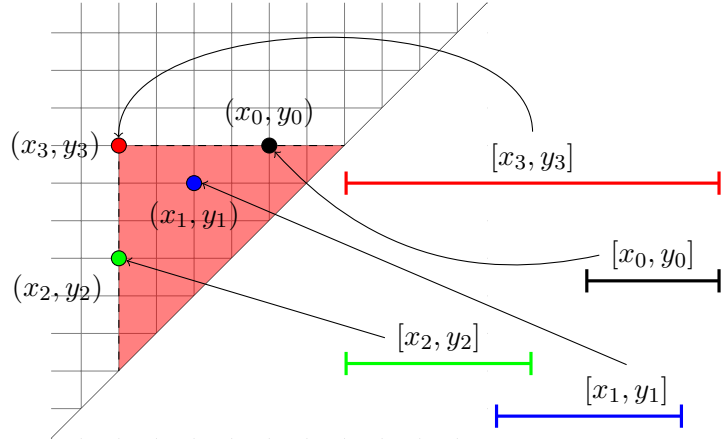
\begin{figure}[tb]
\centering
\begin{tikzpicture}[scale=1]

\draw[step=0.5cm,gray,very thin] (-2.9,-2.9) grid (2.9,2.9);
\fill[color=red, opacity=.5] (-2,1) -- (1,1)-- (-2,-2);
\draw (-2.9,-2.9) -- (2.9,2.9);
\fill[color=white] (-2.9,-2.9) -- (2.9,2.9) -- (2.9,-2.9);
\draw[dashed] (-2,1) -- (-2, -2);
\draw[dashed] (-2,1) -- (1, 1);

\node[shape=circle,draw=black,inner sep=2pt,fill=black, label={above :$(x_0,y_0)$}](A) at (0,1) {};

\node[shape=circle,draw=black,inner sep=2pt,fill=red, label={left :$(x_3,y_3)$}](C) at (-2,1) {};

\node[shape=circle,draw=black,inner sep=2pt,fill=blue, label={below :$(x_1,y_1)$}](B) at (-1,0.5) {};
\node[shape=circle,draw=black,inner sep=2pt,fill=green, label={below left:$(x_2,y_2)$}](D) at (-2,-0.5) {};

\pgftransformshift{\pgfpoint{4cm}{-0.5cm}}

\draw[very thick,|-|] (0.2,-0.3) -- (2,-0.3)node[pos=0.5, above](AI) {$[x_0,y_0]$};

\draw(AI.west) edge[->, bend left] (A);

\draw[very thick,|-|,red] (-3,1) -- (2,1)node[pos=0.5, above=0.001cm,black](CI) {$[x_3,y_3]$};

\draw(CI.north) edge[->, bend right=90, looseness=0.8] (C);

\draw[very thick,|-|,blue] (-1,-2.1) -- (1.5,-2.1)node[pos=0.7, above=0.001cm,black](BI) {$[x_1,y_1]$};

\draw(BI.north) edge[->] (B.east);

\draw[very thick,|-|,green] (-3,-1.4) -- (-0.5,-1.4)node[pos=0.5, above=0.001cm,black](DI) {$[x_2,y_2]$};

\draw(DI.west) edge[->] (D);

\end{tikzpicture}
\caption{\label{fig:compassstructure}Correspondence between intervals and points of the compass structure.}
\end{figure}

As an example, Figure~\ref{fig:compassstructure} shows four intervals
$[x_0,y_0],\ldots,[x_3,y_3]$, respectively represented by the points in the grid $(x_0,y_0),\ldots,(x_3,y_3)$, such that: $(i)$~$[x_0,y_0],\allowbreak [x_1,y_1], [x_2,y_2]\subint [x_3,y_3]$, $(ii)$~$[x_1,y_1] \ssubint\allowbreak [x_3,y_3]$, and $(iii)$~$[x_0,y_0], [x_2,y_2]\not\!\ssubint [x_3,y_3]$.
The red
region highlighted in Figure~\ref{fig:compassstructure} contains all and only those points $(x,y)$ such that $[x,y]\subint[x_3,y_3]$.
Allen interval relation \emph{contains} can thus be represented as a spatial relation between pairs of points. In the following, we make use of $\subint$ and $\ssubint$ also for relating points: given two points $(x,y),(x',y')$ of the grid,
$(x',y')\subint (x,y)$ iff $[x',y'] \subint [x,y]$; similarly, 
$(x',y')\,\ssubint\, (x,y)$ iff $[x',y'] \,\ssubint \,[x,y]$.

Compass structures, that will be repeatedly exploited to solve the satisfiability and model-checking problems for \hsDhom\ and \hsDhomStrict,
can be formally defined as follows.
\begin{defi}\label{def:compassstructure}
Given a 
linear order $\mathbb{S} = \langle S, <\rangle$ and a \hsD-formula (resp., \Dstr-formula) $\varphi$, a  \emph{compass}
$\varphi$-\emph{structure} is a pair $\cG=(\bbP_\bbD,\cL)$,
where $\bbP_\bbD$ is the set of points of the form $(x,y)$,
with $x,y \in S$ and $ x\leq y$, and $\cL$ is a function that
maps any point $(x,y)\in\bbP_\bbD$ to a $\varphi$-atom $\cL(x,y)$
in such a way that for every pair of points $(x,y),(x',y')\in\bbP_\bbD$,
        if $(x',y')\subint (x,y) $ (resp., $(x',y')\,\ssubint\, (x,y)$), then $\cL(x,y) \Dphi \cL(x',y')$
        (\emph{temporal consistency}).

A \emph{weak compass} $\varphi$-\emph{structure} is  a $\varphi$-compass structure where the temporal consistency requirement has been relaxed.
  A (weak) compass $\varphi$-structure $\cG=(\bbP_\bbD,\cL)$ induces the interval model $\bM(\cG)$ over $\mathbb{I(S)}$ whose valuation function  $\cV$
        is defined as follows: for each interval $[x,y]$,   $\cV^{-1}([x,y])=\cL(x,y)\cap \AP$.

\end{defi}
By exploiting Proposition~\ref{prop:characterizeTemporalConsistency} and temporal consistency, we can prove the following lemma, that states an important property of compass structures.
\begin{lem}\label{lem:transitive_req} Let $\varphi$ be a \hsD-formula (resp., \Dstr-formula) and
 $\cG\!=\!(\bbP_\bbD,\cL)$   a compass $\varphi$-structure. $\!\!$ For
all pairs of points $(x'\!,y'),\! (x,y)\!\in\! \bbP_\bbD$,
if $(x',y')\!\subint\! (x,y)$ (resp., $(x',y')\,\ssubint\, (x,y)$), then it holds that $\reqD(\cL(x',y')) \!\subseteq\! \reqD(\cL(x,y))$ and $\obsD(\cL(x',y'))$ $\!\subseteq\! \reqD(\cL(x,y))$.
\end{lem}
%

We now introduce an additional requirement on compass $\varphi$-structures stating that each temporal request is eventually fulfilled.
Formally, \emph{fulfilling}
structures are defined as follows.
\begin{defi}\label{def:fulfillingcompass}
Let $\varphi$ be a \hsD-formula (resp., \Dstr-formula) and $\cG=(\bbP_\bbD,\cL)$   a compass $\varphi$-structure. We say that
$\cG$ is \emph{fulfilling} if for every point $(x,y)\in\bbP_\bbD$ and any formula $\psi\in \reqD(\cL(x,y))$, there exists a point $(x',y')\subint (x,y)$ (resp., $(x',y')\,\ssubint\, (x,y)$) in $\bbP_\bbD$ such that $\psi \in \cL(x',y')$.
\end{defi}
It is worth pointing out that if $\cG$ is fulfilling, then $\reqD(\cL(x,x))=\emptyset$ for all points ``on the diagonal'' $(x,x)\in\bbP_\bbD$
(corresponding to the singleton intervals of $\mathbb{I(S)}$).
%

As proved by Proposition~\ref{prop:compassstructure} below, the fulfillment requirement ensures that, for each point $(x,y)$, the atom $\cL(x,y)$ represents the set of formulas in $\CL(\varphi)$ that hold over the interval $[x,y]$ of the underlying interval model  $\bM(\cG)$.

We say that a compass $\varphi$-structure $\cG=(\bbP_\bbD,\cL)$
\emph{features} a formula $\psi$ if there exists a point $(x,y)\in\bbP_\bbD$
such that $\psi \in \cL(x,y)$.

The next proposition provides a characterization of the set of satisfiable
\hsD-formulas.
%
\begin{prop}\label{prop:compassstructure} Let $\varphi$ be a \hsD-formula (resp., \Dstr-formula). The, the following statements hold:
\begin{enumerate}
  \item given a fulfilling compass $\varphi$-structure $\cG=(\bbP_\bbD, \cL)$, it holds that for all points $(x,y)$  of $\cG$ and $\psi\in \CL(\varphi)$, $\psi\in \cL(x,y)$ if and only if $\bM(\cG),[x,y]\models \psi$;
  \item $\varphi$ is satisfiable if and only if there is a fulfilling compass $\varphi$-structure that
features it.
\end{enumerate}
\end{prop}
\begin{proof} We assume that $\varphi$ is a \hsD-formula (the case where $\varphi$ is a \Dstr-formula is similar). 

We first prove statement~(1) by induction on the structure of the formula $\psi\in\CL(\varphi)$. The base case ($\psi$ is a proposition letter) directly follows from the definition of $\bM(\cG)$. The cases of the Boolean connectives follow from the induction hypothesis and the definition of $\varphi$-atoms. It remains to consider the case where $\psi$ is of the form $\D \psi'$. If $\psi\in \cL(x,y)$, then, being $\cG$ fulfilling, there exists a point $(x',y')$ such that $(x',y')\subint (x,y)$ and
$\psi' \in \cL(x',y')$. By the induction hypothesis, it follows that $\bM(\cG),[x',y']\models \psi'$, hence, $\bM(\cG),[x,y]\models \psi$. As for the opposite direction, assume that $\bM(\cG),[x,y]\models \psi$. Hence, $\bM(\cG),[x',y']\models \psi'$ for some interval $[x',y']$ such that
$[x',y']\subint [x,y]$. By the induction hypothesis, $\psi'\in \cL(x',y')$;  by Lemma~\ref{lem:transitive_req}, we obtain that $\psi\in \cL(x,y)$.

Let us prove now statement~(2). 
First, assume that $\varphi$ is satisfiable. Hence, there exists
an interval model   $\bM$ over $\mathbb{I(S)}$ and an interval $[x,y]\in \mathbb{I(S)}$  such that $\bM,[x,y]\models \varphi$. Let
$\cG=(\bbP_\bbD,\cL)$ be the weak compass $\varphi$-structure where for all points $(x,y)$, $\cL(x,y)$ is the set of formulas $\psi\in\CL(\varphi)$
such that $\bM,[x,y]\models \psi$. Since $\bM,[x,y]\models \varphi$, by the semantics of $\hsD$, it follows that $\cG$ is a fulfilling compass $\varphi$-structure that features $\varphi$. The opposite direction directly follows from statement~(1).
\end{proof}


The notion of homogeneous models directly transfers to compass structures.
\begin{defi}\label{def:hom_compass}
A compass $\varphi$-structure $\cG=(\bbP_\bbD,\cL)$ is \emph{homogeneous}
if for every point $(x,y)\in \bbP_\bbD$ and any $p\in \AP$, $p\in \cL(x,y)$ if and only if $p\in \cL(x',x')$ for all $x'\in [x,y]$.
\end{defi}

Proposition \ref{prop:compassstructure} (item 2) can be tailored to homogeneous compass structures as follows.
\begin{prop}\label{prop:satiffcompass}
A \hsDhom-formula (resp., \hsDhomStrict-formula) $\varphi$ is satisfiable if and only if there is a fulfilling homogeneous compass $\varphi$-structure that
features it.
\end{prop} 

\section{Satisfiability of \hsDhom\ and \hsDhomStrict\ over finite linear orders}\label{sec:decidability}

\newcommand{\genDphi}{\begin{tikzpicture}
\node(A)[inner sep=0pt]{$\scriptstyle\Dphi$};
\node[inner sep=0pt](B) at (0.45,0){};
\node[inner sep=0pt](C) at (-0.4,0){};
\draw[double](C) -- (A);
 \draw[double,->](A) --(B);
\end{tikzpicture}}

In this section, we devise a satisfiability checking procedure for \hsDhom-formulas over finite linear orders, which will also allow us to easily derive a model checking algorithm for \hsDhom{} over finite Kripke structures (see Section~\ref{sec:mc}).
At the end of this section (see Subsection~\ref{sec:SATStrict}), we show how to adapt the proposed approach for \hsDhom\ in order to obtain a decision procedure for satisfiability of \hsDhomStrict\ over finite linear orders.

In the following, we fix a \hsDhom-formula $\varphi$. 
We  first introduce a ternary relation among $\varphi$-atoms, that we denote by $\genDphi$, such that if it holds among all atoms in consecutive positions of a weak compass $\varphi$-structure, then the structure is homogeneous and satisfies both the temporal consistency requirement and the fulfilling one. Hence, we may say that $\genDphi$ is the rule for labeling fulfilling compasses.

Next, we introduce an equivalence relation $\sim$ between \emph{rows} of a compass $\varphi$-structure. Since it has finite index---exponentially bounded by $|\varphi|$---and it preserves fulfillment of compasses, it
makes it possible to ``contract'' the structures when we
identify two related rows.
Moreover, any contraction done according to $\sim$ keeps the same atoms (only the number of their occurrences may vary), and thus if a compass features $\varphi$ before the contraction, then $\varphi$ is still featured after it. This fact is exploited to build a satisfiability checking algorithm for \hsDhom-formulas which makes use of \emph{polynomial working space} only, because $(i)$~it only needs to keep track of two rows of a compass at a time, $(ii)$~all rows satisfy some nice properties that allow one to succinctly encode them, and $(iii)$~compass contractions are implicitly done by means of a reachability check in a suitable graph, whose nodes are the minimal representatives of the  equivalence classes of $\sim$.

\subsection{Labeling of homogeneous fulfilling compasses}
We first show how to label homogeneous fulfilling compass $\varphi$-structures. Such a labeling is based on the aforementioned ternary relation $\genDphi$ among atoms, which is defined as follows.
\begin{defi}\label{def:d_generator}
Given three $\varphi$-atoms $A_1, A_2$ and $A_3$, we say that
$A_3$ is $\Dphi$-generated by $A_1, A_2$
(written $A_1A_2\genDphi A_3$) if:
\begin{itemize}
    \item $A_3\cap\AP = A_1 \cap A_2 \cap \AP$ and
    \item $\reqD(A_3)=\reqD(A_1)\cup \reqD(A_2) \cup \obsD(A_1 )\cup \obsD (A_2)$.
\end{itemize}
\end{defi}

%

It is immediate to show that $A_1A_2\genDphi A_3$ iff $A_2A_1\genDphi A_3$ (i.e., the order of the first two components in the ternary relation is irrelevant).
Intuitively, the first item of the definition enforces the homogeneity assumption.

The next result, which immediately follows from Proposition~\ref{prop:unique}, proves that $\genDphi$ expresses a \emph{functional dependency} on $\varphi$-atoms.
\begin{lem}\label{lem:functional}
Given two $\varphi$-atoms $A_1, A_2\in \Atoms$, there exists \emph{exactly one} $\varphi$-atom  $A_3\in \Atoms$  such that $A_1A_2\genDphi A_3$.
\end{lem}

Definition~\ref{def:d_generator} and Lemma~\ref{lem:functional} can be exploited to label a homogeneous compass $\varphi$-structure $\cG$, namely, to determine the $\varphi$-atoms labeling all the points $(x,y)$ of $\cG$, starting from the ones on the diagonal.
The idea is the following: if two $\varphi$-atoms $A_1$ and $A_2$ label respectively the greatest proper prefix $[x,y-1]$ 
and the greatest proper suffix $[x+1,y]$ 
of the same non-singleton interval $[x,y]$, then the atom $A_3$ labeling point $(x,y)$ is unique, and it is precisely the one satisfying $A_1A_2\genDphi A_3$ (see Figure~\ref{fig:labelling}). The next lemma proves that this is the general rule for labeling homogeneous fulfilling  compasses.

\begin{figure}[tb]
    \centering
    \begin{tikzpicture}[scale=0.8]

\draw[step=1cm,gray, thin] (-3,-3) grid (3,3);

\draw (-3,-3) -- (3,3);
\fill[color=white] (-3,-3) node (v1) {} -- (3.1,3.1) -- (3.1,-3);

\draw [-latex](-3,-3) -- (4,-3);
\draw [-latex](-3,-3) -- (-3,4);

\draw[thick,- triangle 90, decorate, decoration=zigzag] (2,2) -- (-3,2);
\draw[thick,- triangle 90, decorate, decoration=zigzag] (1,1) -- (-3,1);

\node[draw,circle,fill=lightgray] (v3) at (-1,2) {};
\node[draw,circle,fill=lightgray] (v2) at (-1,1) {};
\node[draw,circle,fill=lightgray] (v4) at (0,2) {};


\node (v6) at (-1.5,4) {$(x,y)$};
\node (v5) at (0.5,4) {$(x\!+\!1,y)$};
\node (v7) at (1.8,-0.3) {$(x,y\!-\! 1)$};

\draw [dashed](-1,-1) -- (-1,-3);
\draw [dashed](0,0) -- (0,-3);

\node at (-1,-3.5) {$x$};
\node at (0,-3.5) {$x\!+\!1$};
\node at (-3.5,2) {$y$};
\node at (-3.5,1) {$y\!-\! 1$};

\node at (2,0.9) {$row_{y - 1}$};
\node at (2.8,1.9) {$row_y$};

\draw [gray] (v5) edge (v4);
\draw [gray] (v6) edge (v3);
\draw [gray] (v7.west) edge (v2);

\draw [red,ultra thick,<-] (v3) edge (v2);
\draw [red,ultra thick,<-] (v3) edge (v4);
\end{tikzpicture}
    \caption{Rule for labeling homogeneous fulfilling compass $\varphi$-structures.}
    \label{fig:labelling}
\end{figure}
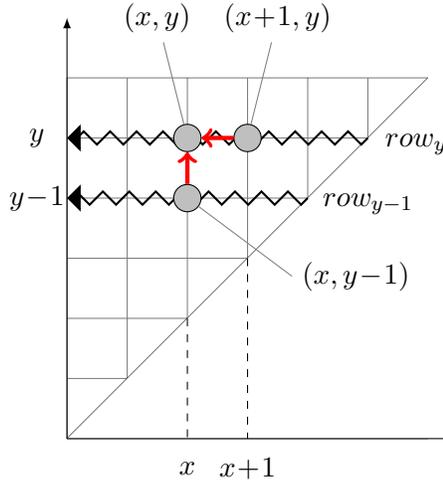

\begin{lem}\label{lem:compass_hom_gen}
Let $\cG=(\bbP_\bbD,\cL)$ be a weak compass $\varphi$-structure. Then, $\cG$ is a homogeneous fulfilling compass $\varphi$-structure \emph{if and only if}  for each point $(x,y) \in \bbP_\bbD$,  it holds that $(i)$~$\cL(x,y-1)\cL(x+1, y)\genDphi \cL(x, y)$ if $x<y$, and $(ii)$~$\reqD(\cL(x,y))=\emptyset$ if $x=y$.
\end{lem}
\begin{proof}
$(\Rightarrow)$
Let us consider a point $(x,y) \in \bbP_\bbD$.
First, we note that, since $\cG$ is fulfilling, it must be $\reqD(\cL(x,y))=\emptyset$ whenever $x=y$.
Otherwise, if $x<y$, we consider the labelings $\cL(x,y-1)$ and $\cL(x+1, y)$. By the homogeneity property
of Definition~\ref{def:hom_compass}, $\cL(x,y)\cap\AP = \cL(x,y-1) \cap \cL(x+1,y)\cap \AP$, and thus condition~$(i)$ of Definition~\ref{def:d_generator} holds.
Moreover,
since $\cG$ is fulfilling, for every $\psi \in \reqD(\cL(x,y))$ we have that either $\psi\in \cL(x,y-1)$, or   $\psi\in \cL(x+1,y)$, or $\psi \in \cL(x',y')$
for some $x<x'\leq y'<y$. In the first two cases, $\psi \in \obsD(\cL(x,y-1)) \cup \obsD(\cL(x+1,y))$. As for the last case,  by Lemma~\ref{lem:transitive_req},  $\obsD(\cL(x',y'))\subseteq \reqD(\cL(x,y-1))$ and $\obsD(\cL(x',y'))\subseteq \reqD(\cL(x+1,y))$,
hence
$\psi \in \reqD(\cL(x,y-1))$ and  $\psi \in \reqD(\cL(x+1,y))$.
We can conclude that $\reqD(\cL(x,y))\subseteq\obsD(\cL(x,y-1)) \cup \obsD(\cL(x+1,y)) \cup  \reqD(\cL(x,y-1)) \cup \reqD(\cL(x+1,y))$. The converse inclusion ($\supseteq$) follows by Lemma~\ref{lem:transitive_req}, hence condition $(ii)$ of Definition~\ref{def:d_generator} holds. This allows us to conclude that $\cL(x,y-1)\cL(x+1, y)\genDphi \cL(x, y)$.

$(\Leftarrow)$ Let us consider a weak compass $\varphi$-structure $\cG=(\bbP_\bbD,\cL)$ such that  for every point $(x,y) \in \bbP_\bbD$, we have $\cL(x,y-1)\cL(x+1, y)\genDphi \cL(x, y)$ if $x<y$, and $\reqD(\cL(x,y))=\emptyset$ if $x=y$. We have to prove that $\cG$ is a homogeneous fulfilling compass $\varphi$-structure.

First, we prove consistency with respect to the relation $\Dphi$.
Let us show that  for all pairs of
points $(x,y)$ and $(x',y')$ with $(x',y')\subint(x,y)$,
we have $\cL(x,y)\Dphi\cL(x',y')$.
The proof is by induction on $\Delta =(x'-x) + (y - y')\geq 1$.
If $\Delta=1$, either $(x',y')=(x+1,y)$
or $(x',y')=(x,y-1)$. Let us consider $(x',y')=(x+1,y)$
(the other case is symmetric). Since
$\cL(x,y-1)\cL(x+1, y)\genDphi \cL(x, y)$, by Proposition~\ref{prop:characterizeTemporalConsistency}
and condition $(ii)$ of Definition~\ref{def:d_generator},
we  get that $\cL(x,y)\Dphi \cL(x+1, y)$.
If $\Delta \geq 2$,
since $(x',y')\subint(x,y)$, then
$(x',y'+1)\subint(x,y)$
or $(x'-1,y')\subint(x,y)$. We only consider
$(x'-1,y')\subint(x,y)$, being the other case symmetric. By the inductive hypothesis,
$\cL(x,y)\Dphi \cL(x'-1,y')$.
Since $\cL(x'-1,y'-1)\cL(x', y')\genDphi \cL(x'-1, y')$,
it holds that $\cL(x'-1, y')\Dphi \cL(x',y')$.
Being $\Dphi$ a transitive relation,
we can conclude that $\cL(x, y)\Dphi \cL(x',y')$.

Let us now show that $\cG$ is fulfilling.
We need to prove that for every point
$(x,y)\in\bbP_\bbD$ and for every
$\psi\in \reqD(\cL(x,y))$, there exists $(x',y')\in\bbP_\bbD$  such that $(x',y')\subint(x,y)$ and
$\psi\in\cL(x',y')$. The proof is by induction on
$y-x\geq 0$. If $x=y$, we have $\reqD(\cL(x,y))=\emptyset$, hence the thesis holds vacuously.
If $y-x\geq 1$, since
$\cL(x,y-1)\cL(x+1, y)\genDphi \cL(x, y)$,
we have $\reqD(\cL(x, y))= \reqD(\cL(x,y-1))\cup\reqD(\cL(x+1, y))\cup \obsD(\cL(x,y-1))\cup\obsD(\cL(x+1, y))$. If $\psi \in \obsD(\cL(x,y-1))\cup\obsD(\cL(x+1, y))$, the thesis is verified. If $\psi \in \reqD(\cL(x+1, y))$ (the case $\psi \in \reqD(\cL(x, y-1))$ is symmetric and thus omitted), by the inductive hypothesis,
$\psi \in \cL(x'',y'')$ for some
$(x'',y'') \subint (x+1, y) \subint (x,y)$.

It remains to prove that $\cG$
is homogeneous. We have to show that for every $(x,y)\in\bbP_\bbD$, $\cL(x,y)\cap \AP= \bigcap_{x'\in [x,y]}\cL(x',x')\cap \AP$. The proof is by induction
on $y-x\geq 0$. If $x=y$, the property trivially holds. Let us assume now $y-x>0$ (inductive step).
Since $\cL(x+1,y)\cL(x,y-1)\genDphi \cL(x, y)$, by condition $(i)$ of Definition~\ref{def:d_generator} and the induction hypothesis,
we obtain that $\cL(x,y)\cap \AP= \bigcap_{x'\in [x+1,y]}\cL(x',x')\cap \bigcap_{x'\in [x,y-1]}\cL(x',x')\cap \AP$. Hence, the result directly follows.
 \end{proof}


\subsection{The contraction method}
In this section, we describe the proposed contraction method. To begin with,
we introduce the notion of \emph{$\varphi$-row}, which can be viewed as the ordered sequence of (the occurrences of) atoms labelling a row of a compass $\varphi$-structure. 

For a non-empty finite word (or sequence) $w$ over some finite alphabet $\Sigma$, we denote by $|w|$ the length of $w$. Moreover, for all $0\leq i <|w|$, $w[i]$ denotes the $(i+1)^{th}$ letter of $w$.

Given two non-empty finite words $w,w'$ over $\Sigma$, we denote by $w\cdot w'$ the concatenation of   $w$ and $w'$.
If the last letter of $w$ coincides with the first letter of $w'$,  we denote by $w\star w '$ the word  $w \cdot w'[1]\ldots w'[n-1] $, where $n=|w'|$, that is, the word obtained by concatenating $w$ with the word obtained from $w'$ by erasing the first letter. When $|w'|=1$, $w\star w'=w$.

%
\begin{defi}\label{def:row}
A \emph{$\varphi$-row $\row$} is a non-empty finite sequence of $\varphi$-atoms such that for all  $0\leq i<|\row|-1$, $\row[i+1]  \Dphi \row[i]$ and $(\row[i]\cap \AP)\supseteq (\row[i+1]\cap \AP)$. We say that the $\varphi$-row $\row$ is \emph{initialized} if $\reqD(\row[0])=\emptyset$.
We represent a   $\varphi$-row $\row$ in the form  $\row=A_0^{m_0}\cdots A_n^{m_n}$ (\emph{maximal factorization}), where $A^m$ stands for $m$ repetitions of the $\varphi$-atom $A$, $m_i>0$ for all $i\in [0,n]$, and $A_i\neq A_{i+1}$ for all $i\in [0,n-1]$.
\end{defi}

\newcommand{\rownext}{\begin{tikzpicture}
\node(A)[inner sep=0pt]{$\scriptstyle row_{\varphi}$};
\node[inner sep=0pt](B) at (0.6,0){};
\node[inner sep=0pt](C) at (-0.5,0){};
\draw[double](C) -- (A);
 \draw[double,->](A) --(B);
\end{tikzpicture}}

Let $\Rows$ be the set of all possible $\varphi$-rows. This set is infinite.
The next lemma proves that 
the number of distinct atoms in any $\varphi$-row $\row=A_0^{m_0}\cdots A_n^{m_n}$ is linearly bounded in the size of $\varphi$.

\begin{lem}\label{lem:compas_implies_rowA} 
The number of distinct atoms in a $\varphi$-row $\row$ is at most $2|\varphi|$. Moreover, if $A_0^{m_0}\cdots A_n^{m_n}$ is the maximal factorization of $\row$, then $A_0,\ldots,A_n$ are pairwise distinct.
\end{lem}
\begin{proof}
First of all, since for each $0\leq i < n$, $A_{i+1} \Dphi A_i$,
it holds that $\reqD(A_i) \subseteq \reqD(A_{i+1})$.
Therefore, two monotonic sequences can be associated with every $\varphi$-row, one increasing ($\reqD(A_0) \subseteq \reqD(A_1)\subseteq \ldots \subseteq \reqD(A_n)$),
and one decreasing ($(A_0 \cap \AP)\supseteq (A_1\cap \AP)\supseteq \ldots \supseteq (A_n\cap \AP)$).
The number of distinct elements is bounded by $|\varphi|$ in the former sequence and by $|\varphi|+1$ in the latter, as $|\REQ_{\varphi}|\leq |\varphi|-1$ and $|\AP|\leq |\varphi|$, since, w.l.o.g., we can restrict ourselves to the proposition letters actually occurring in $\varphi$. Given that, as already shown (Proposition~\ref{prop:unique}), a set of requests and a set of proposition letters uniquely determine a $\varphi$-atom, any $\varphi$-row may feature at most $2|\varphi|$ distinct atoms, that is, $n\leq 2|\varphi|$. The proof of the second statement is immediate.
\end{proof}

Given a homogeneous compass $\varphi$-structure $\cG= (\bbP_\bbD,\cL)$ (for $\bbD =(S,<)$), for every  $y\in S$, we define $\row_y$ as the word of $\varphi$-atoms $\row_y=\cL(y,y)\cdots \cL(0, y)$, that is, the sequence of atoms labeling points of $\cG$ with the same $y$-coordinate, starting from the one on the diagonal inwards (see Figure~\ref{fig:labelling}). Since in a fulfilling compass $\varphi$-structure there are no temporal requests in the atoms labeling the diagonal points, we obtain the following result.

\begin{lem}\label{lem:compas_implies_rowB} 
Let $\cG= (\bbP_\bbD,\cL)$  (for $\bbD =(S,<)$) be a  homogeneous fulfilling compass $\varphi$-structure. For every $y\in S$, $\row_y$ is an initialized $\varphi$-row.
\end{lem}
%

We now define the \emph{successor function} over $\varphi$-rows,  which, given a $\varphi$-row $\row$ and a \mbox{$\varphi$-atom} $A$, returns the $\varphi$-row of length $|\row|+1$ and first atom $A$ obtained by a component-wise application of $\genDphi$ starting from $A$ and the first atom of $\row$. 
\begin{defi}\label{def:rownext}
%
 Given a $\varphi$-atom $A$ and a  $\varphi$-row
$\row$ with $|\row| =n$, the \emph{$A$-successor of $\row$}, denoted by $\SUCC(\row,A)$, is the sequence $B_0\ldots B_{n}$ of $\varphi$-atoms defined as follows:
$B_0= A$ and $\row[i] B_i \genDphi B_{i+1}$ for all $i\in [0,n-1]$.
\end{defi}

By Proposition~\ref{prop:characterizeTemporalConsistency} and Definition~\ref{def:d_generator}, we deduce the following lemma.
\begin{lem}\label{lem:propert_successor} The following properties hold:
\begin{enumerate}
  \item Let $\row$ be a   $\varphi$-row and $A$ be a $\varphi$-atom. Then, $\SUCC(\row,A)$ is a  $\varphi$-row.
  \item Let $\row$ be a $\varphi$-row of the form $\row = \row_1\cdot \row_2$ and  $A$ be a $\varphi$-atom.
  Then, $\SUCC(\row,A)=\SUCC(\row_1,A)\star \SUCC(\row_2,A_1)$, where $A_1$ is the last $\varphi$-atom of $\SUCC(\row_1,A)$.
\end{enumerate}
\end{lem}
\begin{proof} Property~(2) directly follows  from Definition~\ref{def:rownext}. As for Property~(1), let $\SUCC(\row,A)$ $= B_0\ldots B_{n}$, where $n= |\row|$.
By Definitions~\ref{def:d_generator} and~\ref{def:rownext}, for all $i\in [0,n-1]$, $\reqD(B_i)\subseteq \reqD(B_{i+1})$, $\obsD(B_i)\subseteq \reqD(B_{i+1})$, and $(B_i\cap \AP)\supseteq (B_{i+1}\cap \AP)$. By
   Proposition~\ref{prop:characterizeTemporalConsistency}, it holds that for all $i\in [0,n-1]$, $B_{i+1}  \Dphi B_i$ and $(B_i\cap \AP)\supseteq (B_{i+1}\cap \AP)$.
   Hence, $\SUCC(\row,A)$ is a  $\varphi$-row.
\end{proof}

Moreover, by Lemma~\ref{lem:compass_hom_gen},  consecutive  rows in homogeneous fulfilling compass $\varphi$-structures respect the successor function. In particular, the next result directly follows from Lemmata~\ref{lem:compass_hom_gen} and~\ref{lem:compas_implies_rowB}.
\begin{lem}\label{lem:row_successor}
Let $\cG=(\bbP_\bbD,\cL)$ be a weak compass $\varphi$-structure such that $\reqD(\cL(x,x))=\emptyset$ for all $(x,x)\in\bbP_\bbD$. Then, $\cG$ is a  homogeneous fulfilling compass $\varphi$-structure
\emph{if and only if} for each $0\leq y < |S| -1$, $\row_{y+1}=\SUCC(\row_y,\row_{y+1}[0])$. 
\end{lem}
%


We now illustrate the kernel of the proposed approach to solve satisfiability for  \hsDhom-formulas.
We introduce an equivalence relation $\sim$ of finite index over $\Rows$
whose number of classes is singly exponential in the size of $\varphi$ and such that each  class has a
representative whose length is polynomial in the size of $\varphi$. As a crucial result, we show that the successor function preserves the
equivalence between $\varphi$-rows.

 The equivalence relation $\sim$ is based on the notion of \emph{rank of $\varphi$-atoms}.
Given an atom $A \in \Atoms$, we define the \emph{rank of $A$}, written $\rank(A)$, as 
$|\REQ_\varphi| - |\reqD(A)|$.  Clearly, $0\leq \rank(A)< |\varphi|$. Whenever $A \Dphi A'$, for some $A' \in \Atoms$, $\reqD(A')\subseteq \reqD(A)$, and hence  $\rank(A)\leq \rank(A')$. 
We can see the $\rank$ of an atom as the ``number of degrees of freedom''
that it gives to 
the atoms  that stay ``above it''.
In particular, by Definition~\ref{def:row}, for every  $\varphi$-row $\row=A_0^{m_0} \cdots A_n^{m_n}$, we have $\rank(A_0)\geq \cdots \geq  \rank(A_n)$.
%
%
\begin{defi}\label{def:equivalence_class}
Given two   $\varphi$-rows $\row_1=A_0^{m_0} \cdots A_n^{m_n}$ and $\row_2=
\hA_0^{\hm_0} \cdots \hA_{\hn}^{\hm_{\hn}}$ (represented in maximal factorization form), we say that they are \emph{equivalent}, written
$\row_1 \sim \row_2$,
if $(i)$~$n=\hn$, and $(ii)$~for each $0\leq i \leq n$, $A_i =\hA_i$, and \emph{either} ${m_i=\hm_i}$ \emph{or}
both $m_i$ and $\hm_i$ are (strictly) greater than $\rank(A_i)$.

A \emph{minimal $\varphi$-row} is a $\varphi$-row whose maximal factorization $A_0^{m_0} \cdots A_n^{m_n}$ satisfies the following condition: $m_i\in [1, \rank(A_i) + 1]$, for each $0\leq i\leq n$.
\end{defi}

Note that if two rows feature the same set of atoms, the lower the rank of an atom
$A_i$, the lower the number of occurrences of $A_i$ both the rows have to feature in order to belong to the same equivalence class.
As an example, let $\row_1$ and $\row_2$ be two rows with $\row_1=A_0^{m_0}A_1^{m_1}$,
$\row_2=A_0^{\om_0}A_1^{\om_1}$, $\rank(A_0) = 4$,
and $\rank(A_1)= 3$. If $m_1 = 4$ and $\om_1=5$ they are both greater than $\rank(A_1)$, and
hence they do not violate the condition for $\row_1 \sim \row_2$. On the other hand,
if $m_0 = 4$ and $\om_0=5$, we have that $m_0$  is less than or equal to $\rank(A_0)$. Thus, in this case, $\row_1 \not\sim \row_2$ due to the indexes of $A_0$. This happens because $\rank(A_0)$ is greater than $\rank(A_1)$. Two cases in which $\row_1 \sim \row_2$
are $m_0 = \om_0$ and $m_0, \om_0 \geq 5$.

The next result directly follows from the definition of 
$\sim$ and Lemma~\ref{lem:compas_implies_rowA}.

\begin{lem}\label{lem:minimal_representative}
Each equivalent class of $\sim$ contains a unique minimal $\varphi$-row. Moreover, the length of a minimal $\varphi$-row is at most $O(|\varphi|^{2})$.
\end{lem}

Thus, the equivalence relation $\sim$ has finite index given by the number of minimal $\varphi$-rows. This number is roughly bounded  by  the number of  all the possible $\varphi$-rows $\row=A_0^{m_0}\cdots A_n^{m_n}$, with exponents $m_i$ ranging from  $1$ to $ |\varphi|$. Since $(i)$~the number of possible $\varphi$-atoms is $2^{|\varphi|}$, $(ii)$~the number of \emph{distinct} atoms in any $\varphi$-row is at most $2|\varphi|$, and $(iii)$~the number of possible functions
$f: \{1,\ldots , \ell\} \rightarrow \{ 1,\ldots ,|\varphi|\}$ is $|\varphi|^{\ell}$, we have that the number of distinct equivalence classes of $\sim$ is bounded by
\[
    \sum_{j=1}^{2|\varphi|} (2^{|\varphi|})^j\cdot |\varphi|^j\leq  2^{3|\varphi|^2},
\]
which is exponential in the length of the input formula
$\varphi$.


Next, we observe that if we replace a segment (sub-row) of a  $\varphi$-row with an equivalent one, we obtain a  $\varphi$-row
which is equivalent to the original one. The following lemma holds.

\begin{lem}\label{lem:row_concatenation} Let $\row_1,\row'_1,\row_2,\row'_2$ be $\varphi$-rows such that
$\row_1 \!\sim\! \row'_1$ and \mbox{$\row_2 \!\sim\! \row'_2$}. If $\row_1\star\row_2$ and $\row'_1\star\row'_2$ are defined, then
  $\row_1\star \row_2 \sim \row'_1\star \row'_2$.
\end{lem}
\begin{proof}  Let $A$ be the first common $\varphi$-atom $A$ of $\row_2$ and $\row'_2$. By hypothesis, $A$ is also the last common atom
of $\row_1$ and $\row'_1$.
By hypothesis and Definition~\ref{def:equivalence_class}, one can easily deduce that, representing by $m$ (resp., $m'$) the number of occurrences of $A$ in $\row_1\star \row_2$ (resp., $\row'_1\star \row'_2$), it holds that either $m=m'$, or both of them are greater than $\rank(A)$.
Hence, the result easily follows.
\end{proof}

We now show that the successor function on  $\varphi$-rows preserves the equivalence of  $\varphi$-rows.
We first show that the result holds for $\varphi$-rows of the form $B^{m}$ for some $\varphi$-atom
$B$ and $m\geq 1$.

\begin{lem}\label{lem:succ_preserves_prelimin}
Let $A$ and $B$ be two $\varphi$-atoms and $m>\rank(B)$. Then, the following properties hold:
\begin{itemize}
  \item  the  $\varphi$-row $\SUCC(B^{m},A)$ is of the form
$A A_1\ldots A_{k}^{\ell}$ for some $k,\ell\geq 1$  such that $A_1,\ldots,A_{k}$ are pairwise
distinct $\varphi$-atoms  and $\ell>\rank(A_k)$ (note that $A$ and $A_1$ may be equal),
  \item for each $t>0$, $\SUCC(B^{m+t},A) = \SUCC(B^{m},A)\cdot A_k^{t}$.
\end{itemize}
\end{lem}
\begin{proof}
Let $\row$ be the $\varphi$-row of length $m+1$ given by $\SUCC(B^{m},A)$. Since $\row[0]=A$ and $\row[i] B \genDphi \row[i+1]$ for all $i\in [0,m-1]$,
by Definition~\ref{def:d_generator}, we easily deduce that for all $i\in [1,m]$, the following conditions hold:
\begin{itemize}
  \item $\row[i] \cap \AP = B\cap A \cap \AP$;
  \item $\rank(\row[i-1])\geq \rank(\row[i])$ and $\rank(B)\geq \rank(\row[i])$;
  \item if $i<m$, then $\row[i]\neq \row[i+1]$ if and only if $\rank(A_i)>\rank(A_{i+1})$;
  \item if $i<m$ and $\row[i]= \row[i+1]$, then $\row[j]=\row[i]$ for all $j\geq i$.
\end{itemize}

Since $\row[0]=A$ and $|\row|=m+1$, it follows that $\row$ is of the form $A A_1\ldots A_{k}^{\ell}$ for some $k,\ell\geq 1$ such that $k+\ell-1= m$, $A_1,\ldots,A_{k}$ are pairwise
distinct $\varphi$-atoms, and $\rank(B)\geq \rank(A_1)>\cdots >\rank(A_k)$. Hence, $\rank(B)\geq \rank(A_k) +k-1$.
By hypothesis,  $m>\rank(B)$ which entails that $\ell=m-k+1>\rank(A_k)$, and the first statement of Lemma~\ref{lem:succ_preserves_prelimin} follows.

As for the second statement, we have that $\SUCC(B^{m+t},A) = \SUCC(B^{m},A)\star  \SUCC(B^{t},A_k)$ (by Lemma~\ref{lem:propert_successor}(2)). If $\ell=1$, then, being
$0\leq \rank(A_k)<\ell$, it holds that $\rank(A_k)=0$. Hence, we deduce that
each atom occurring in $\SUCC(B^{t},A_k)$ is $A_k$. On the other hand, if $\ell>1$ then $A_k B \genDphi A_k$. Hence, in both the cases, we obtain that
$\SUCC(B^{t},A_k)$ is $A_k^{t+1}$ and $\SUCC(B^{m+t},A) = \SUCC(B^{m},A)\cdot A_k^{t}$, which concludes the proof.
\end{proof}

We now generalize Lemma~\ref{lem:succ_preserves_prelimin} to arbitrary $\varphi$-rows.

\begin{lem}\label{lem:succ_preserves_equiv}
Let $A$ be a $\varphi$-atom and $\row$, $\row'$ be two  $\varphi$-rows such that  $\row \sim \row'$.
Then,
it holds that $\SUCC(\row,A)\sim \SUCC(\row',A)$.
\end{lem}
\begin{proof} The proof is by induction on the number of distinct $\varphi$-atoms occurring in $\row$, denoted by $N(\row)$. Being $\row$ and $\row'$ equivalent, $N(\row')=N(\row)$.

\smallskip

\noindent \emph{Base case:} $N(\row)=N(\row') =1$. Assume that $|\row|\leq |\row'|$ (the case where $|\row'|\leq |\row|$ being symmetric).
Since $\row$ and $\row'$ are equivalent, there is a $\varphi$-atom $B$ such that
$\row =B^{m}$, $\row' = B^{m+t}$,   $m=|\row|$, $t=|\row|-|\row'|$, and either $t=0$ or $m>\rank(B)$. If $t=0$, that is, $\row=\row'$, the result is obvious.
Otherwise, the result directly follows from Lemma~\ref{lem:succ_preserves_prelimin}.

\smallskip

\noindent \emph{Inductive step:} $N(\row)=N(\row') >1$. Hence, being $\row \sim \row'$, $\row$ (resp., $\row'$) can be written in the form   $\row=\row_1\cdot \row_2$ (resp., $\row'=\row'_1\cdot \row'_2$) such that $\row_1\sim \row'_1$, $\row_2\sim \row'_2$, $N(\row_1)=N(\row'_1)<N(\row)=N(\row')$,
and $N(\row_2)=N(\row'_2)<N(\row)=N(\row')$. Let $A_1$ (resp., $A'_1$) be the last atom in $\SUCC(\row_1,A)$ (resp., $\SUCC(\row'_1,A)$). By the induction hypothesis,
$\SUCC(\row_1,A)\sim \SUCC(\row'_1,A) $, $A_1=A'_1$, and $\SUCC(\row_2,A_1)\sim \SUCC(\row'_2,A'_1) $.
By Lemma~\ref{lem:propert_successor}(2), $\SUCC(\row,A)=\SUCC(\row_1,A)\star \SUCC(\row_2,A_1)$ and $\SUCC(\row',A)=\SUCC(\row'_1,A)\star \SUCC(\row'_2,A'_1)$.
By applying Lemma~\ref{lem:row_concatenation}, we have that $\SUCC(\row,A)\sim \SUCC(\row',A)$ proving the thesis.
\end{proof}

\subsection{A satisfiability checking procedure for \hsDhom}
Let us now focus on the complexity of the
satisfiability checking problem for a \hsDhom-formula $\varphi$ over finite linear orders, which has been proved, by Proposition~\ref{prop:satiffcompass}, to be equivalent to the problem of deciding whether there is a homogeneous fulfilling compass $\varphi$-structure that features $\varphi$.
By exploiting Lemma~\ref{lem:succ_preserves_equiv}, we reduce such a problem to a reachability problem in a finite graph with the initialized minimal $\varphi$-rows as vertices.

\begin{defi}\label{def:equivalencegraph}
Let $\row$ be a minimal $\varphi$-row and $A$ 
an atom. We denote by $\SuccMIN(\row,A)$ the unique minimal $\varphi$-row in the equivalence class of $\sim$ containing  $\SUCC(\row,A)$.  

We associate with formula $\varphi$ the finite graph
$G_{\varphi}^{min}=(\RowsMIN,\RelMIN)$ defined as:
\begin{itemize}
  \item $\RowsMIN$ is the set of initialized minimal $\varphi$-rows;
  \item for all $\row,\row'\in \RowsMIN$, $\row\,\RelMIN\,\row'$  iff  $\row' = \SuccMIN(\row,\row'[0])$.
\end{itemize}
\end{defi}

We now prove the main technical result of the section.

 \begin{thm}\label{thm:path_iff_sat}
Let $\varphi$ be a \hsDhom-formula. Then, there exists a homogeneous fulfilling compass $\varphi$-structure $\cG=(\bbP_\bbD, \cL)$ that features $\varphi$
\emph{if and only if} there exist two initialized minimal $\varphi$-rows $\row_1$ and $\row_2$ such that:
\begin{enumerate}
    \item  $|row_1|=1$,  $\varphi\in \row_2[i]$ for some $0\leq i<|\row_2|$, and
    \item $\row_2$ is reachable from $\row_1$ in the finite graph $G_{\varphi}^{min}=(\RowsMIN,\RelMIN)$.
\end{enumerate}
\end{thm}
\begin{proof}
($\Rightarrow$)
Let us consider a homogeneous fulfilling compass $\varphi$-structure $\cG=(\bbP_\bbD, \cL)$ that features $\varphi$.
By Lemma~\ref{lem:compas_implies_rowB} and Lemma~\ref{lem:row_successor}, there is $m\geq 0$ and
$m+1$ initialized \mbox{$\varphi$-rows} $\row_0,\ldots,\row_m$ such that $|\row_0|=1$, $\varphi\in \row_m[i]$ for some $0\leq i<|\row_m|$,
and $\row_{i+1}=\SUCC(\row_i,\row_{i+1}[0])$ for all $0\leq i<m$.

For each $0\leq i\leq m$, let $\row_i^{min}$ be the unique minimal $\varphi$-row in the equivalence class $[\row_i]_{\sim}$. Note that
$\row_0^{min}=\row_0$, $\row_i^{min}$ is initialized, $\row_i^{min}[0]=\row_i[0]$, for all $i\in [0,m]$, and $\varphi\in \row_m^{min}[j]$, for some $0\leq j<|\row^{min}_m|$.
Moreover, by Lemma~\ref{lem:succ_preserves_equiv}, $\SUCC(\row_i^{min},\row_{i+1}[0])$ is equivalent to $\row_{i+1}=\SUCC(\row_i,\row_{i+1}[0])$,
for all $0\leq i<m$. By the definition of $\SuccMIN$, 
 $\row^{min}_{i+1}=\SuccMIN(\row^{min}_i,\row_{i+1}[0])$, for all $0\leq i<m$. Hence, $\row_m^{min}$ is reachable from  $\row_0^{min}$ in the finite graph $G_{\varphi}^{min}$, and thus the thesis.\vspace{0.1cm}

($\Leftarrow$) Let us assume that there exist  two initialized minimal $\varphi$-rows $\row_1$ and $\row_2$
such that $|row_1|=1$,  $\varphi\in \row_2[i]$ for some $0\leq i<|\row_2|$, and $\row_2$ is reachable from $\row_1$ in the finite graph $G_{\varphi}^{min}$.
Hence, there is $m\geq 0$ and $m+1$ initialized minimal $\varphi$-rows $\row_0^{min},\ldots,\row_m^{min}$ such that $\row_0^{min}=\row_1$,
$\row_m^{min}=\row_2$,
and $\row^{min}_{i+1}=\SuccMIN(\row^{min}_i,\row^{min}_{i+1}[0])$, for all $0\leq i<m$. Let $\row'_0,\ldots,\row'_m$ be the sequence of $\varphi$-rows
defined as follows: $\row'_0=\row_0^{min}=\row_1$ and  $\row'_{i+1}=\SUCC(\row'_i,\row^{min}_{i+1}[0])$, for all $0\leq i<m$. By applying
Lemma~\ref{lem:succ_preserves_equiv}, we deduce that $\row'_i\sim \row_i^{min}$, for all $i\in [0,m]$. Hence, $\row'_i$ is initialized for all $i\in [0,m]$,
and $\varphi\in \row'_m[j]$, for some $0\leq j< |\row'_m|$.
Now, let $\cG=(\bbP_\bbD, \cL)$ be a weak compass $\varphi$-structure, with $S= \{0,\ldots, m\}$ and
$\cL(x,y)=\row_y'[y-x]$, for every $0\leq x\leq y\leq m$.
By Lemma~\ref{lem:row_successor}, $\cG$ is a  homogeneous fulfilling compass $\varphi$-structure that features $\varphi$. This concludes the proof. 
\end{proof}

\begin{algorithm}[tb]
\begin{algorithmic}[1]
\caption{
\hfill
\parbox[t]{0.83\linewidth}{
\texttt{SAT}$(\varphi)$ \hfill Input: a \hsDhom-formula $\varphi$\newline
\null\hfill \emph{\small Non-deterministic procedure deciding the satisfiability of a \hsDhom-formula $\varphi$}
}
}\label{NDAlgo}
    \State{$M\gets 2^{3|\varphi|^2}$, $step\gets 0$ and $\row\gets A$ for some atom $A\in \Atoms$ with  $\reqD(A)=\emptyset$}
    \If{there exists $0\leq i<|\row|$ such that $\varphi\in \row[i]$}
        \State{\textbf{return} ``satisfiable''}
    \EndIf
    \If{$step = M-1$}
        \State{\textbf{return} ``unsatisfiable''}
    \EndIf
    \State{Guess an atom $A\in \Atoms$ with  $\reqD(A)=\emptyset$ and set $\row'=\SuccMIN(\row,A)$}
    \State{$step \gets step +1$ and $\row \gets\row'$}
    \State{Go back to line 2}
\end{algorithmic}
\end{algorithm}

The size of $\RowsMIN$ is bounded by $M=2^{3|\varphi|^2}$. By Theorem~\ref{thm:path_iff_sat}, $\varphi$ is satisfiable if and only if there is path in the
finite graph  $G_{\varphi}^{min}=(\RowsMIN,\RelMIN)$ of length at most $M$ from a $\varphi$-row  in $\RowsMIN$   of length $1$ to a $\varphi$-row
 $\row_2\in \RowsMIN$ such that $\varphi\in \row_2[i]$, for some $0\leq i<|\row_2|$.
The \emph{non-deterministic} procedure \texttt{SAT}$(\varphi)$ in Algorithm~\ref{NDAlgo} exploits such a characterization to decide the satisfiability
of a \hsDhom-formula $\varphi$. Initially, the algorithm guesses a  $\varphi$-atom having no temporal request, that is, a row in $\RowsMIN$ of length~$1$.
At the $j$-th iteration, if the currently processed $\varphi$-row $\row\in \RowsMIN$ has some atom which contains $\varphi$, then the algorithm  terminates with success. Otherwise, the algorithm guesses a successor $\row'$ of the current $\varphi$-row $\row$ in $G_{\varphi}^{min}$.
The procedure terminates after at most $M$ iterations. 
The working space used by the procedure is polynomial:
$M$
and $step$ (which ranges in $[0,M-1]$) can be encoded
in binary with $\lceil \log_2 M \rceil +1=O(|\varphi|^2)$ bits. Moreover,
at each step, the algorithm keeps in memory only two minimal initialized $\varphi$-rows: the current one $\row$ and the guessed successor $\row'$ in
$G_{\varphi}^{min}$. By Lemma~\ref{lem:minimal_representative}, each minimal initialized $\varphi$-row can be represented by using space polynomial in $\varphi$. Thus, since $\NPsp=\Psp$, we obtain the following result.

%
%
\begin{thm}\label{thm:pspace}
The satisfiability problem for \hsDhom-formulas over finite linear orders is in $\Psp$.
\end{thm}

\subsection{The easy adaptation to \hsDhomStrict{}}\label{sec:SATStrict}

We conclude the section by sketching the changes to the previous notions that allow us to prove the decidability of the satisfiability problem for \hsDhomStrict{} over finite linear orders. 
As a  matter of fact, it suffices to replace the definitions
of $\genDphi$,  $\varphi$-$\row$, and $\SUCC$
by the following ones.
For the sake of simplicity, we introduce a dummy atom $\boxdot$, for which we assume $\reqD(\boxdot)=\obsD(\boxdot)=\emptyset$.

\begin{defi}\label{def:d_generatorS}
Given the $\varphi$-atoms $A_1, A_3,A_4\in\Atoms$ and $A_2\in\Atoms\cup\{\boxdot\}$, we say that
$A_4$ is \mbox{$\Dphi\ssubint$-generated} by $A_1,A_2,A_3$
(written $A_1,A_2, A_3\genDphiS A_4$) if:
\begin{itemize}
    \item $A_4\cap\AP = A_1 \cap A_3 \cap \AP$ and
    \item $\reqD(A_4)=\reqD(A_1)\cup \reqD(A_3) \cup \obsD(A_2 )$.
\end{itemize}
\end{defi}
The intuition behind Definition \ref{def:d_generatorS} is that if an interval $[x,y]$, with $x<y$, is labeled by $A_4$, and its three sub-intervals $[x,y-1]$, $[x+1,y-1]$, and $[x+1,y]$ are labeled by $A_1,A_2$, and $A_3$, respectively, we require that $A_1,A_2, A_3\genDphiS A_4$ holds. In particular, in the special case where $x=y-1$, we have $A_2=\boxdot$ since $[x+1,y-1]$ is not a valid interval. Moreover, since the only strict sub-interval is $[x+1,y-1]$ (i.e., $[x+1,y-1]\,\ssubint\,[x,y]$), we require that  $\obsD(A_2)\subseteq\reqD(A_4)$. Finally, since the requests of $A_1$ and $A_3$ are fulfilled by a strict sub-interval of $[x,y]$, we require that $\reqD(A_1)\subseteq\reqD(A_4)$ and $\reqD(A_3)\subseteq\reqD(A_4)$.

\begin{defi}\label{def:rowS}
A \emph{$\varphi$-$\ssubint$-row $\row$} is a non-empty finite sequence of $\varphi$-atoms such that for all  $0\leq i<|\row|-1$,
$\reqD(\row[i]) \subseteq \reqD (\row[i+1])$
 and $(\row[i]\cap \AP)\supseteq (\row[i+1]\cap \AP)$. The $\varphi$-$\ssubint$-row $\row$ is \emph{initialized} if $\reqD(\row[0])=\emptyset$.

\end{defi}

\begin{defi}\label{def:rownextS}
Given a $\varphi$-atom $A$ and a  $\varphi$-$\ssubint$-row
$\row$, with $|\row| =n$, the \emph{$A$-$\ssubint$-successor of $\row$}, denoted by $\SUCCStrict(\row,A)$, is the sequence $B_0\ldots B_{n}$ of $\varphi$-atoms defined as follows:
$B_0= A$ and $\row[i]\row[i-1] B_i \genDphi B_{i+1}$ for all $i\in [0,n-1]$, 
with $\row[i-1]=\boxdot$ for $i=0$.
\end{defi}

Once the above-defined changes have been to the basic notions, following exactly the same steps
of the proof of Theorem~\ref{thm:pspace}, we can show that like \hsDhom{},  satisfiability for \hsDhomStrict{} over finite linear orders is in $\Psp$.

\begin{thm}\label{thm:SATStrict}
The satisfiability problem for \hsDhomStrict-formulas over finite linear orders is in $\Psp$.
\end{thm}

$\Psp$-completeness of the satisfiability problem for \hsDhom{} and \hsDhomStrict{}  
will be proved in Section~\ref{sec:hardness}. In the next section, we focus on
their model checking problem. 

\section{Model checking of \hsDhom{} and \hsDhomStrict{} formulas over Kripke structures} \label{sec:mc}

In this section, we deal with the \emph{model checking} 
problem for \hsDhom\ and \hsDhomStrict, namely, the problem of checking whether  a model of a given system satisfies some behavioural properties expressed as \hsDhom- or \hsDhomStrict-formulas.
The usual models are \emph{Kripke structures},
which will now be introduced along with the definition of the semantics of \hsDhom{} and \hsDhomStrict{} formulas over them.

\begin{defi}
A \emph{finite Kripke structure} is a tuple $\Ku=(\AP,\States, \Edges,\mu,\sinit)$, where $\AP$ is a finite set of proposition letters, $\States$ is a finite set of states,
$\Edges\subseteq \States\times \States$ is a left-total binary relation over $\States$,
$\mu :\States\to 2^\AP$ is a  labelling function over $\States$, and $\sinit\in \States$ is the initial state.
\end{defi}

For all $s\in \States$, $\mu(s)$ is the set of proposition letters that hold on $s$,
while $\Edges$ is the transition relation that describes the evolution of the system over time.

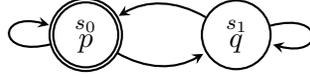
\begin{figure}[tb]
\centering
\begin{tikzpicture}[->,>=stealth,thick,shorten >=1pt,auto,node distance=2cm,every node/.style={circle,draw}]
    \node [style={double}](v0) {$\stackrel{s_0}{p}$};
    \node (v1) [right of=v0] {$\stackrel{s_1}{q}$};
    \draw (v0) to [bend right] (v1);
    \draw (v1) to [bend right] (v0);
    \draw (v0) to [loop left] (v0);
    \draw (v1) to [loop right] (v1);
\end{tikzpicture}
\caption{Kripke structure $\mathpzc{K}_2$.}\label{KEquiv}
\end{figure}

Figure~\ref{KEquiv} depicts a finite Kripke structure $\mathpzc{K}_2 = (\{p,q\},\{s_0,s_1\},\Edges,\mu,\sinit)$,
with
$\Edges=\{(s_0,s_0), (s_0,s_1), (s_1,s_0), (s_1,s_1)\}$,
$\mu(s_0)=\{p\}$, and $\mu(s_1)=\{q\}$.
The initial state $\sinit$ is identified by a double circle.

\begin{defi}[Paths and traces]
Given  a finite Kripke structure $\Ku=(\AP,\States, \Edges,\mu,\sinit)$, a \emph{path of $\Ku$} is a non-empty finite sequence of states
$\rho =s_1\cdots s_n$  such that $(s_i,s_{i+1})\in \Edges$ for $i = 1,\ldots ,n-1$. A path is \emph{initial} if it starts from the initial state
of $\Ku$.

We extend the labeling $\mu$ to paths of $\Ku$ in the usual way: for a path $\rho=s_1\ldots s_n$, $\mu(\rho)$ denotes the word over $2^{\AP}$
of length $n$ given by $\mu(s_1)\ldots \mu(s_n)$. A \emph{trace} of $\Ku$ is a non-empty finite word  over $2^{\AP}$ of the form $\mu(\rho)$ for some path
$\rho$ of $\Ku$. A trace is \emph{initial} if it is of the form  $\mu(\rho)$ for some initial path $\rho$ of  $\Ku$.
\end{defi}

%

Given a non-empty finite word $w$ over $2^{\AP}$, we can associate with $w$,  in a natural way,  a homogeneous interval model $\bM(w)$ over the finite linear order induced by $w$.

\begin{defi}\label{def:inducedmodel}
For a non-empty finite word $w$ over $2^{\AP}$,
the interval model
$\bM(w)=\langle\mathbb{I(S)},\cV\rangle$ \emph{induced by $w$} is the homogeneous interval model defined as follows:
\begin{enumerate}
    \item $\mathbb{S} = \langle S, <\rangle$, where $S=\{0,\ldots,|w|-1\}$, and
    \item for every interval $[x,y]\in \mathbb{I(S)}$ and $p\in\AP$, $[x,y]\in\cV(p)$ if and only if $p\in w[x']$ for~all~$x'\in [x,y]$.
\end{enumerate}
\end{defi}

\begin{defi}[Model checking of \hsDhom{} and \hsDhomStrict{} formulas] Let $\varphi$ be \hsDhom-formula (resp., \hsDhomStrict-formula).
Given a non-empty finite word $w$ over $2^{\AP}$,   \emph{$w$ satisfies $\varphi$}, denoted by $w\models \varphi$, if $\bM(w),[0,|w|-1]\models\varphi$.
A finite Kripke structure $\Ku$ over $\AP$ is \emph{a model of the formula $\varphi$}  if
for each initial trace $w$ of $\Ku$, it holds that $w\models \varphi$.
%
The \emph{model-checking} (\emph{MC}) \emph{problem}  for \hsDhom{} (resp., \hsDhomStrict{}) is
the problem of deciding for a given finite Kripke structure $\Ku$ and \hsDhom-formula (resp., \hsDhomStrict-formula) $\varphi$,  whether $\mathpzc{K}\models \varphi$.
\end{defi}


\begin{exa}\label{exKrSched}

\begin{figure}[tb]
\centering
\begin{tikzpicture}[->,>=stealth',shorten >=1pt,auto,node distance=2.2cm,thick,main node/.style={circle,draw}]

  \node[main node,style={double}] (1) {$\stackrel{v_0}{\emptyset}$};
  \node[main node,fill=gray!35] (3) [below=0.7cm of 1] {$\stackrel{v_2}{p_2}$};
  \node[main node,fill=gray!50] (2) [left of=3] {$\stackrel{v_1}{p_1}$};
  \node[main node,fill=gray!20] (4) [right of=3] {$\stackrel{v_3}{p_3}$};
  \node[main node,fill=gray!50] (5) [below of=2] {$\stackrel{\overline{v_1}}{p_1}$};
  \node[main node,fill=gray!35] (6) [below of=3] {$\stackrel{\overline{v_2}}{p_2}$};
  \node[main node,fill=gray!20] (7) [below of=4] {$\stackrel{\overline{v_3}}{p_3}$};

  \path[every node/.style={font=\small}]
    (1) edge [bend right] node[left] {$r_1$} (2)
        edge node {$r_2$} (3)
        edge [bend left] node[right] {$r_3$} (4)
    (2) edge [bend right] node [left] {$u_1$} (5)
    (3) edge node {$u_2$} (6)
    (4) edge [bend left] node [right] {$u_3$} (7)
    (5) edge node[very near end,left] {$r_2$} (3)
    (5) edge [out=270,in=260,looseness=1.3] node [near start,swap] {$r_3$} (4)
    (6) edge node[very near end,right] {$r_1$} (2)
    (6) edge node[very near end,left] {$r_3$} (4)
    (7) edge [out=270,in=280,looseness=1.3] node [near start] {$r_1$} (2)
    (7) edge node[very near end,right] {$r_2$} (3)
    ;
\end{tikzpicture}
\vspace{-1cm}
\caption{The Kripke structure $\mathpzc{K}_{Sched}$ of Example~\ref{exKrSched}.}\label{KSched}
\end{figure}
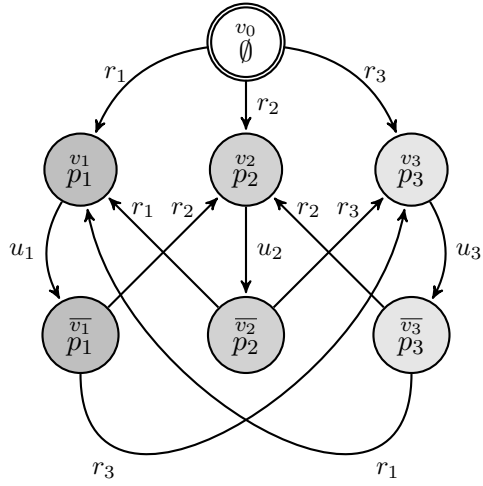

In Figure~\ref{KSched}, we give an example of a finite Kripke structure $\mathpzc{K}_{Sched}$ that models the behaviour of a scheduler serving three processes which are continuously requesting the use of a common resource. The initial state
is $v_0$: no process is served in that state. In the states $v_i$ and $\overline{v}_i$, with $i \in \{1,2,3\}$, the $i$-th process is served (this is denoted by the fact that $p_i$ holds in those states). For the sake of readability, edges are marked either by $r_i$, for $request(i)$, or by $u_i$, for $unlock(i)$. Edge labels do not have a semantic value, that is, they are neither part of the structure definition, nor proposition letters; they are simply used to ease reference to edges.
Process $i$ is served in state $v_i$, then, after ``some time'', a transition $u_i$ from $v_i$ to $\overline{v}_i$ is taken; subsequently, process $i$ cannot be served again immediately, as $v_i$ is not directly reachable from $\overline{v}_i$ (the scheduler cannot serve the same process twice in two successive rounds). A transition $r_j$, with $j\neq i$, from $\overline{v}_i$ to $v_j$ is then taken and process $j$ is served. This structure can be easily generalised to an arbitrary number of processes.

We now show how some meaningful behavioural properties of
the Kripke structure $\mathpzc{K}_{Sched}$
can be expressed by \hsDhom-formulas. 

Preliminarily, we observe that the formula $len_{\geq i} := \D^{i-1} \top$ is satisfied by traces at least $i$ units long, and, analogously, $len_{\leq i} := \bD^i \bot$ by traces at most $i$ units long. We define $len_{=i}:=len_{\geq i} \wedge len_{\leq i}$.

In all the following formulas, we force the validity of the considered properties over all legal computation sub-intervals by using the modality $\bD$ (all computation sub-intervals are sub-intervals of at least one initial trace of the Kripke structure).

The first formula requires that \lq\lq at least 2 processes are witnessed in any sub-interval of length at least 5 of an initial trace\rq\rq. Since a process cannot be executed twice in a row, such a formula is satisfied by $\mathpzc{K}_{Sched}$:
\[
\Ku_{Sched}\models\bD\big( len_{\geq 5} \rightarrow \bigvee_{1\leq i<j \leq 3} (\D p_i \; \wedge \; \D p_j)\big).
\]

The second formula requires that \lq\lq in any sub-interval of length at least 11 of an initial trace, process 3 is executed at least once in some states\rq\rq{} (\emph{non starvation}). $\mathpzc{K}_{Sched}$ does not satisfy it, because the scheduler can postpone the execution of a process ad libitum:
\[
\Ku_{Sched}\not\models \bD\big( len_{\geq 11} \rightarrow \D p_3 \big).
\]

The third formula requires that \lq\lq in any sub-interval of length at least 6 of an initial trace, $p_1,p_2$, and $p_3$ are all witnessed\rq\rq. The only way to satisfy this property would be to force the scheduler to execute the three processes in a strictly periodic manner (\emph{strict alternation}), that is, $p_i p_j p_k p_i p_j p_k p_i p_j p_k\cdots$, for $i,j,k \in \{1,2,3\}$ and $i \neq  j \neq  k \neq i$, but $\mathpzc{K}_{Sched}$ does not meet such a
requirement:
\[
\Ku_{Sched}\not\models \bD\big(len_{\geq 6}\rightarrow (\D p_1\; \wedge \; \D p_2\; \wedge \D p_3 )\big).
\]

Finally, we write two formulas expressing \emph{safety} properties: \lq\lq it is never the case that processes 1 and 2 are executed consecutively\rq\rq, and \lq\lq it is never the case that a state where process 1 is executed, is reached\rq\rq. Neither of these is satisfied:
\[
\Ku_{Sched}\not\models \bD\big(len_{=4}\rightarrow (\neg\D p_1 \; \vee \neg\D p_2)\big),
\]
\[
\Ku_{Sched}\not\models\bD \neg p_1.
\]
\end{exa} 

We now show how, by slightly modifying the satisfiability checking procedure given in Section~\ref{sec:decidability}, it is possible to obtain an automata-theoretic MC algorithm for \hsDhom-formulas  over finite Kripke structures $\Ku$.
Let $G_{\varphi}^{min}=(\RowsMIN,\RelMIN)$ be the finite graph  of Definition~\ref{def:equivalencegraph} associated with the \hsDhom-formula $\varphi$. We first show that it is possible to construct a standard  deterministic finite automaton (DFA) $\tilde{\mathcal{N}}_\varphi$ over the alphabet $2^{\AP}$ with set of states $\RowsMIN$, which accepts all and only the non-empty finite words over $2^{\AP}$ that satisfy formula $\varphi$.
Next, given a finite Kripke structure $\Ku$ and a \hsDhom-formula $\varphi$, to check whether $\Ku$ is a model of $\varphi$, we apply the standard approach to MC taking the synchronous product  of $\Ku$ with the automaton $\tilde{\mathcal{N}}_{\neg\varphi}$ for the negation of the formula $\varphi$ ($\Ku\times \tilde{\mathcal{N}}_{\neg\varphi}$ for short).
$\Ku\times \tilde{\mathcal{N}}_{\neg\varphi}$ accepts all and only the initial traces of $\Ku$ that violate the property $\varphi$. Hence,
$\Ku$ is a model of $\varphi$ if and only if the language accepted by $\Ku\times \tilde{\mathcal{N}}_{\neg\varphi}$ is empty. 

We now provide the technical details.

A nondeterministic finite-state automaton (NFA) is a tuple $\mathcal{N} =(\Sigma,Q,q_1,\delta,F)$, where $\Sigma$ is a finite alphabet, $Q$ is a finite set of states, $q_1\in Q$ is the \emph{initial} state, $\delta: Q\times \Sigma \to  2^Q$ is the transition function, and $F\subseteq Q$ is the set of \emph{accepting} states. Given a finite word $w$ over $\Sigma$, with $|w|=n$, a computation of $\mathcal{N}$ over $w$ is a finite sequence of states $q_1',\ldots,q_{n+1}'$ such that $q_1'=q_1$, and for all $i\in [0,n-1]$, $q_{i+2}' \in \delta(q_{i+1}',w[i])$. The language $\mathcal{L}(\mathcal{N})$  accepted by $\mathcal{N}$ consists of the finite words $w$ over $\Sigma$ such that there is a computation over $w$ ending in some accepting state.
A deterministic finite-state automaton (DFA) is an NFA $\tilde{\mathcal{N}} =(\Sigma,\tilde{Q},\tilde{q_1},\tilde{\delta},\tilde{F})$ such that for all $(q,c)\in \tilde{Q}\times \Sigma$, $\tilde{\delta}(q,c)$ is a singleton.

Let  $\Ku=(\AP,\States, \Edges,\mu,\sinit)$ be a finite Kripke structure and $\mathcal{N} =(2^{\AP},Q,q_1,\delta,F)$ be an NFA\@. The \emph{synchronous product}  of $\Ku$ and $\mathcal{N}$ (denoted by $\Ku\times \mathcal{N}$) is the NFA $(2^{\AP},\States\times Q,(\sinit,q_1),\delta', \States\times F)$, where for all $(s,q)\in \States\times Q$ and $P\in 2^{\AP}$, $\delta'((s,q),P)=\emptyset$ if $P\neq \mu(s)$, and $\delta'((s,q),P)$ is the set of pairs $(s',q')\in \States \times Q$
such that $(s,s')\in \Edges$ and $q'\in \delta(q,P)$ otherwise. 
It can be easily seen that $\Ku\times \mathcal{N}$ accepts all and only the initial traces of $\Ku$ which are accepted by $\mathcal{N}$.

Let $\varphi$ be a \hsDhom-formula and $G_{\varphi}^{min}=(\RowsMIN,\RelMIN)$ be the finite graph of Definition~\ref{def:equivalencegraph}, where $\RowsMIN$ is the set of initialized minimal $\varphi$-rows and $\row\,\RelMIN\,\row'$  if and only if  $\row' = \SuccMIN(\row,\row'[0])$.
\begin{defi}\label{def:automatonForD}
Let $P\subseteq \AP$ be a set of proposition letters and $\varphi$ be a \hsDhom-formula. We denote by $A(P)$ the unique $\varphi$-atom such that $A(P)\cap\AP=P$ and $\reqD(A(P))=\emptyset$.
We associate with $\varphi$ the DFA $\tilde{\mathcal{N}}_\varphi =(2^{\AP},\RowsMIN\cup \{\tilde{q_1}\},\tilde{q_1},\tilde{\delta},\tilde{F})$, where $\tilde{\delta}$ and $\tilde{F}$ are defined as follows:
\begin{itemize}
\item $\tilde{\delta}(\tilde{q_1},P)= A(P)$, for all $P\in 2^{\AP}$;
\item $\tilde{\delta}(\row, P) = \SuccMIN(\row,A(P))$, for all $P\in 2^{\AP}$ and $\row\in \RowsMIN$;
\item $\tilde{F}$ is the set of $\varphi$-row $\row\in \RowsMIN$ such that $\varphi\in \row[n-1]$, where $n=|\row|$.
\end{itemize}
\end{defi}

By exploiting Theorem~\ref{thm:path_iff_sat}, we get the following result, that outlines an automata-theoretic approach to MC for \hsDhom.

  \begin{thm}\label{thm:AutomatonFormula}
Let $\varphi$ be a \hsDhom-formula. Then, the DFA $\tilde{\mathcal{N}}_\varphi$ accepts all and only the non-empty finite words over $2^{\AP}$ that satisfy $\varphi$.
\end{thm}
\begin{proof}
Let $w$ be a non-empty finite word over $2^{\AP}$ and $n=|w|-1$. We show that $\bM(w),[0,n]\models\varphi$ if and only if $w\in\mathcal{L}(\tilde{\mathcal{N}}_\varphi)$.

($\Rightarrow$)
Let us assume that $\bM(w),[0,n]\models\varphi$. We define $\cG$ as the weak compass $\varphi$-structure $(\bbP_\bbD,\cL)$, where $\bbD =(\{0,\ldots,n\},<)$ and, for all points $(x,y)$, $\cL(x,y)$ is the set of formulas $\psi\in\CL(\varphi)$ such that $\bM(w),[x,y]\models \psi$.
By the semantics of \hsDhom{}, $\cG$ is a homogeneous fulfilling compass $\varphi$-structure. For all $i\in [0,n]$, let $\row_i$ be the $i^{th}$ $\varphi$-row of  $\cG$. By construction, $\varphi\in \row_n[n]$ and $\row_i[0]=A(w[i])$, for all $i\in [0,n]$. By the proof of Theorem~\ref{thm:path_iff_sat} (right implication), there exist  $n+1$ initialized  minimal $\varphi$-rows $\row_0^{min},\ldots,\row_n^{min}$ such that $\row_0^{min}=\row_0$, $\row_i^{min}$ is equivalent to $\row_i$, for all $i\in [0,n]$, and $\row^{min}_{i+1}=\SuccMIN(\row^{min}_i,\row^{min}_{i+1}[0])$, for all $i\in [0,n-1]$.
It follows that the last atom of $\row_n^{min}$ contains $\varphi$ and $\row_i^{min}[0]=A(w[i])$, for all $i\in [0,n]$. By Definition~\ref{def:automatonForD}, it follows that there is an accepting computation of $\tilde{\mathcal{N}}_\varphi$ over $w$, that is, $w\in\mathcal{L}(\tilde{\mathcal{N}}_\varphi)$.

($\Leftarrow$)
Let us assume that $w$ is accepted by $\tilde{\mathcal{N}}_\varphi$. By Definition~\ref{def:automatonForD}, there are $n+1$ initialized  minimal $\varphi$-rows $\row_0^{min},\ldots,\row_n^{min}$ such that $\row_0^{min}=A(w(0))$, $\varphi \in \row_n^{min}[n]$, $\row_i^{min}[0]=A(w(i))$, for all $0\leq i\leq n$,  and $\row^{min}_{i+1}=\SuccMIN(\row^{min}_i,\row^{min}_{i+1}[0])$ for all $0\leq i<n$. By the proof of Theorem~\ref{thm:path_iff_sat} (left implication), there is a fulfilling homogeneous compass $\varphi$-structure  $\cG=(\bbP_\bbD, \cL)$, with $S= \{0,\ldots, n\}$, such that, for all $0\leq j<n$, its $j^{th}$ row $\row_j$ is equivalent to $\row_j^{min}$.
Hence, $\row_i[0]=A(w(i))$, for all $0\leq i\leq n$, and the last atom of $\row_n$ contains $\varphi$. Since in a homogeneous interval model the valuation function is completely specified by the values taken at the singleton intervals, it follows that $\bM(\cG)=\bM(w)$. Moreover, by Proposition~\ref{prop:compassstructure}, $\bM(\cG),[0,n]\models\varphi$. Hence,
   $\bM(w),[0,n]\models\varphi$ and the thesis follows.
\end{proof}

By Theorem~\ref{thm:AutomatonFormula}, we get the main result of the section.

\begin{thm}\label{thm:pspaceMC}
The MC problem for \hsDhom-formulas (resp., \hsDhomStrict{}-formulas) over finite linear orders is in $\Psp$. 
For constant-length formulas, it is in $\NLOGSP$.
\end{thm}
\begin{proof}
 By Theorem~\ref{thm:AutomatonFormula},  given a finite Kripke structure $\Ku$ and a \hsDhom-formula $\varphi$,
$\Ku \not\models \varphi$ if and only if
the language accepted by
$\Ku\times \tilde{\mathcal{N}}_{\neg\varphi}$
is not empty.
Similarly to Algorithm~1 of Section~\ref{sec:decidability}, the problem of establishing whether $\mathcal{L}(\Ku\times \tilde{\mathcal{N}}_{\neg\varphi})\neq \emptyset$ can be solved by a nondeterministic algorithm which uses space logarithmic in the number of states of $\Ku\times \tilde{\mathcal{N}}_{\neg\varphi}$  and checks whether some accepting state is reachable from the initial one. Since the number of states in $\Ku\times \tilde{\mathcal{N}}_{\neg\varphi}$ is linear in the number of states of $\Ku$ and singly exponential in the length  of $\varphi$, and the complexity classes $\NPsp=\Psp$ and $\NLOGSP$ are closed under complementation, the result for \hsDhom{} directly follows. The MC procedure for \hsDhom{} can be easily adapted to \hsDhomStrict{} by making use of Definitions~\ref{def:d_generatorS}--\ref{def:rownextS}.
\end{proof}


%
 In the next section, we will prove that MC for \hsDhom-formulas and \hsDhomStrict-formulas  is $\Psp$-hard. 

\section{Hardness of MC and satisfiability checking of \hsDhom{} and \hsDhomStrict{} formulas over finite linear orders}\label{sec:hardness}

In this section, we provide lower bounds for MC and satisfiability checking for
\hsDhom{} (resp., \hsDhomStrict{}) over finite linear orders that match the upper bounds of Theorem~\ref{thm:pspaceMC} and
Theorem~\ref{thm:pspace} (resp., Theorem~\ref{thm:SATStrict}). 

 \begin{thm}\label{theorem:hardness}
MC and satisfiability checking for \hsDhom-formulas (resp., \hsDhomStrict-formulas) over finite linear orders  are both $\Psp$-hard.
\end{thm}

By a trivial reduction from the \emph{problem of \mbox{(non-)reachability}} of two nodes in a directed~graph, it easily follows that MC for \emph{constant-length} \hsDhom-formulas (resp., \hsDhomStrict{}-formulas) is $\NLOGSP$-hard. By taking into account the upper bounds given by Theorems~\ref{thm:pspaceMC}, \ref{thm:pspace}, and~\ref{thm:SATStrict},  we obtain the following corollary.

\begin{cor}
MC and satisfiability checking for \hsDhom-formulas (resp., \hsDhomStrict-formulas) over finite linear orders  are both  $\Psp$-complete. Moreover, when the length of
the formula is fixed, MC is $\NLOGSP$-complete.
\end{cor}

The rest of the section is devoted to the proof of Theorem~\ref{theorem:hardness}. We focus on  \hsDhom. The proof for \hsDhomStrict{} is very similar, and thus we omit it.

We prove Theorem~\ref{theorem:hardness}   by means of a polynomial time reduction from a domino-tiling problem for grids with rows of linear length~\cite{harel92}. To start with,
we fix an instance $\Instance$ of such a problem, which is a tuple $\Instance =\tupleof{C,\Delta,n,d_\Init,d_\Final}$, where $C$ is a finite set of colors, $\Delta \subseteq C^{4}$ is a set of tuples $\tupleof{c_\Down,c_\Left,c_\Up,c_\Right}$ of four colors, called \emph{domino-types}, $n>1$ is a  natural number encoded in \emph{unary},
and $d_\Init,d_\Final\in\Delta$ are two distinguished domino-types (respectively, the initial and final ones).
A \emph{grid} of $\Instance$ is a mapping $f: [1, \ell]\times[1,n]\mapsto  \Delta$  for
some natural number $\ell>0$. Note that each row of a grid consists of $n$ cells and each cell contains a domino type.
A \emph{\emph{tiling} of $\Instance$} is a grid $f: [1, \ell]\times[1,n]\mapsto  \Delta$ satisfying the following additional requirements.
\begin{itemize}
  \item two adjacent cells in a row have the same color on the shared edge, namely, for all $(i,j)\in [1,\ell]\times [1,n-1]$,
   $[f(i,j)]_{\Right}=[f(i,j+1)]_{\Left}$ (\emph{row constraint});
  \item two adjacent cells in a column have the same color on the shared edge, namely, for all $(i,j)\in [1,\ell-1]\times [1,n]$,
   $[f(i,j)]_{\Up}=[f(i+1,j)]_{\Down}$ (\emph{column constraint});
  \item $f(1,1)=d_\Init$ (\emph{initialization}) and $f(\ell,n)=d_\Final$ (\emph{acceptance}).
\end{itemize}

\begin{rem}\label{remark:hardness}
Without loss of generality, we can assume that for all domino-types $d,d'\in\Delta$:
\begin{itemize}
  \item if $[d]_{\Right}=[d']_{\Left}$, then $d\neq d_\Final$ and $d'\neq d_\Init$;
 \item if $[d]_{\Up}=[d']_{\Down}$, then $d\neq d_\Final$ and $d'\neq d_\Init$.
\end{itemize}
This ensures that the first cell of a tiling $f$ is the only one containing $d_\Init$, and the last cell of $f$ is the only one containing $d_\Final$.
\end{rem}

It is well known that the problem of checking the existence of a tiling of $\Instance$ is $\Psp$-complete~\cite{harel92}.
To prove the statement about MC of Theorem~\ref{theorem:hardness}, we show 
 how to construct, in polynomial time, a finite Kripke structure $\Ku_\Instance$ and a \hsDhom{} formula
$\Phi_\Instance$ such that there is a tiling of $\Instance$ if and only if $\Ku_\Instance\not\models \neg\Phi_\Instance$ (Subsection~\ref{sec:hardnessMC}).
As for the claim 
about satisfiability checking, we explain  how to construct, in polynomial time  for the given instance $\Instance$, a \hsDhom{} formula which is satisfiable if and only if there is a tiling of $\Instance$ (Subsection~\ref{sec:hardnessSAT}).

First, we define a suitable encoding of the tilings of $\Instance$ by non-empty finite words over $2^{\AP}$, where
the set $\AP$ of atomic propositions is given by  $\AP = \Delta \times [1,n]\times \{0,1\}$. In the following, we identify non-empty finite words $w$ over $2^{\AP}$ with the induced homogeneous interval models $\bM(w)$.

We encode the row of a tiling
by concatenating the codes of the row’s cells starting from the first cell, and by marking the encoding with a tag which is a bit in $\{0,1\}$. Each cell code keeps track of the associated content and position along the row. Formally, a \emph{row-code} with tag $b\in\{0,1\}$   is a word over $2^{\AP}$ of length $n$ having the form
$
\{(d_1,1,b)\} \ldots \{(d_n,n,b)\}
$
 such that the following holds:
\begin{itemize}
  \item for all $i\in [1,n-1]$, $[d_i]_{\Right}=[d_{i+1}]_{\Left}$ (\emph{row constraint}).
\end{itemize}
 A sequence $\nu$ of row-codes is \emph{well-formed} if for each non-last row-code in $\nu$ with tag
$b$, the next row-code in $\nu$ has tag $1-b$ for all $b\in \{0,1\}$ (i.e., the tag changes in moving from a row-code to the next one).
Tilings $f$  are then encoded by  words over $2^{\AP}$ corresponding to well-formed concatenations of the codes
of the rows of $f$ starting from the first row.

Note that by adding two additional padding colors $c_L$ and $c_R$ and domino types whose up and down parts are distinct from $c_L$ and $c_R$, we can also assume that each \emph{partial} row-code $\{(d_i,i,b)\} \ldots \{(d_j,j,b)\}$, where $i\leq j$, $d_i\neq d_\Init$ if $i\neq 1$, and $d_j\neq d_\Final$ if $j\neq n$,  can always be extended to a whole row-code.

\subsection{$\Psp$-hardness for MC against \hsDhom{}}\label{sec:hardnessMC} Let as define an \emph{initialized} well-formed sequence of row-codes as a sequence whose first symbol is $\{(d_\Init,1,0)\}$). By construction, the following result
trivially holds.

\begin{prop}\label{prop:KripkeStrucutureHardness}  A finite Kripke structure $\Ku_\Instance$ over $2^{\AP}$ such that the initial traces $w$ of $\Ku$ correspond to the non-empty prefixes of the initialized well-formed sequences of row-codes can be built in time polynomial in the size of $\Instance$.
\end{prop}

We now construct in polynomial time  a \hsDhom{} formula $\Phi_\Instance$ such that, given a non-empty prefix $w$ of an initialized well-formed sequence of row-codes, $w$ is a model of  $\Phi_\Instance$ if and only if $w$ has some prefix encoding a tiling of $\Instance$. Hence, by Proposition~\ref{prop:KripkeStrucutureHardness}, there exists a tiling of $\Instance$ if and only if there is some initial trace $w$ of $\Ku_\Instance$ such that $\Ku_\Instance,w\models \Phi_\Instance$ if and only if $\Ku_\Instance\not\models \neg \Phi_\Instance$.

Let $w$ be a non-empty prefix of an initialized well-formed sequence of row-codes. Formula $\Phi_\Instance$ must enforce the acceptance requirement and the column constraint between adjacent row-codes. For acceptance, it suffices to force  $w$ to visit some position where proposition $(d_\Final,n,b)$ holds for some tag $b\in\{0,1\}$. This can be done with the following \hsDhom{} formula: 
\[
\varphi_\Final := \displaystyle{\bigvee_{b\in\{0,1\}}} \D  (d_\Final,n,b)
\]
For the column constraint, it suffices to ensure the following requirement.
\begin{description}
  \item[(C-Req)] for each infix $\eta$ of $w$ of length $n+1$, let $\ell$ and $\ell'$ be the unique two distinct positions
  of $\eta$, with $\ell$ preceding $\ell'$, with the same cell index. Then, it holds that $[d_\ell]_\Up = [d_{\ell'}]_\Down$, where $d_\ell$ (resp., $d_{\ell'}$) is the domino-type of position $\ell$ (resp., $\ell'$).
\end{description}
In order to express this requirement in \hsDhom{}, we crucially observe that positions $\ell$ and $\ell'$ of the infix $\eta$ of length $n+1$ have distinct tag. Moreover, $\ell$ precedes $\ell'$ if and only if \emph{either} (i) $\ell$ and $\ell'$ have cell index
$n$ and there is a position of $\eta$ with cell index $1$ and the same tag as $\ell'$, \emph{or} (ii) $\ell$ and $\ell'$ have cell index $i$, for some  $i<n$, and there is a position of $\eta$ with cell index $i+1$ and the same tag as $\ell$.  Requirement~(C-Req) can thus be expressed by the following \hsDhom{} formula: 
\[
\begin{array}{ll}
   \varphi_{C} := & \bD\bigl(len_{=n+1} \longrightarrow \\
   & \displaystyle{\bigvee_{d,d'\in\Delta\mid [d]_\Up = [d']_\Down}\,\,\bigvee_{b\in\{0,1\}}}\bigl[
   (\D(d,n,b)\wedge \D(d',n,1-b)\wedge \displaystyle{\bigvee_{d''\in\Delta}}(d'',1,1-b))\,\vee
     \\
    & \phantom{\displaystyle{\bigvee_{d,d'\in\Delta\mid [d]_\Up = [d']_\Down}}\bigl[} \displaystyle{\bigvee_{i\in[1,n-1]}} (\D(d,i,b)\wedge \D(d',i,1-b)\wedge \displaystyle{\bigvee_{d''\in\Delta}}(d'',i+1,b))\,\bigr]\bigr)
\end{array}
\]
where $len_{=n+1}$ is the formula $\D^{n} \top\wedge \bD^{n+1} \bot$ capturing the infixes of length $n+1$. The desired
formula $\Phi_\Instance$ is given by
$\varphi_\Final\wedge \varphi_C$. This concludes the proof of the statement about MC of 
Theorem~\ref{theorem:hardness} (assuming the $\subint$-semantics).

\subsection{$\Psp$-hardness for \hsDhom{} satisfiability checking}\label{sec:hardnessSAT}

For an instance $\Instance$ of the considered domino-tiling problem, we construct in polynomial time a
\hsDhom{} formula $\Psi_\Instance$ (under the $\subint$-semantics) such that $\Instance$ has some tiling if and only if $\Psi_\Instance$ is satisfiable.

To build the formula $\Psi_\Instance$, we exploit as conjuncts the two \hsDhom{} formulas
$\varphi_\Final$ and $\varphi_C$, ensuring the acceptance requirement and the
column constraint, respectively, given in Subsection~\ref{sec:hardnessMC}.  Additionally, we introduce as a third conjunct in the definition of $\Psi_\Instance$, the \hsDhom{} formula $\varphi_{wf}$, which, intuitively, ``emulates'' the behaviour of the Kripke structure $\Ku_\Instance$ in Proposition~\ref{prop:KripkeStrucutureHardness}. Thus, $\Psi_\Instance:= \varphi_{wf}\wedge \varphi_\Final \wedge \varphi_C$.

The models of $\varphi_{wf}$ are all and only the non-empty finite words $w$ over $2^{\AP}$ satisfying the following condition: either $w$ or its reverse $w^{R}$ coincides with the infix $\eta$ of some well-formed sequence  of row-codes such that (i) $\eta$ visits some cell code
with cell index $1$ and content $d_\Init$, and (ii) $\eta$ contains some row-code. By construction (see Subsection~\ref{sec:hardnessMC}), in both cases ($w=\eta$ or $w^{R}=\eta$) the conjunct $\varphi_C$ in the definition of
$\Psi_\Instance$ enforces the column constraint on $\eta$. Moreover,
  our assumption on the instance $\Instance$ (see Remark~\ref{remark:hardness}) entails that the
$d_\Final$-positions in $\eta$ follow the $d_\Init$-positions. This ensures that (i) for each model $w$ of
$\psi_{\Instance}$, either $w$ or its reverse $w^{R}$ has an infix  encoding a tiling  of $\Instance$, and (ii) each tiling of $\Instance$ is a model of $\Psi_\Instance$.

We now provide the technical details of the construction of $\varphi_{wf}$, which is the conjunction of the following subformulas (w.l.o.g., we assume $n>2$): 
\begin{enumerate}
  \item \emph{Propositional mutual exclusion}: at every position, at most one proposition letter holds. This is expressed  by the \hsDhom formula
  \[
\neg \displaystyle{\bigvee_{p,p'\in\AP\mid p\neq p'}} \D (p\wedge p')
  \]
  \item \emph{Initialization}: there is some position where $(d_\Init,1,b)$ holds for some tag $b\in\{0,1\}$. This is captured by the in \hsDhom{} formula
    \[
     \displaystyle{\bigvee_{b\in\{0,1\}}} \D  (d_\Init,1,b)
    \]
  \item \emph{$n$-length constraint}: there is some infix of the given word $w$ ($w$ included) having length
  $n$ and such that all its positions have the same tag. This requirement can be formalized in \hsDhom{} as $\D \psi \vee \psi$,  where $\psi$ is the \hsDhom{} formula
    \[
\psi:= len_{=n} \wedge \displaystyle{\bigvee_{b\in \{0,1\}}} \bD (len_{=1} \rightarrow  \displaystyle{\bigvee_{d\in \Delta}\bigvee_{i\in [1,n]}}(d,i,b))
    \]
     \item \emph{Fixed ordering of cell indexes}: for each infix $\eta$ of the given word $w$ ($w $ included) having length
  $n$ and for each cell index $i\in [1,n]$, there is some position of $\eta$ with cell index $i$. This requirement ensures that there is a bijection (permutation) $\pi$ over $[1,n]$ such that the projection of $w$ over the $[1,n]$-component of $\AP$ has the form
    \[
    \pi(i+1)\ldots \pi(n) \,[\pi(1)\ldots \pi(n)]^{*}\, \pi(1)\ldots \pi(j-1) \text{ for some }i,j\in [1,n]
    \]
    Such a requirement can be expressed as $\bD \xi \wedge \xi$ where $\xi$ is the \hsDhom{} formula
    \[
\xi:= len_{=n}  \longrightarrow  \displaystyle{\bigwedge_{i\in [1,n]}} \D  \displaystyle{\bigvee_{d\in \Delta}\bigvee_{b\in \{0,1\}}}(d,i,b)
    \]
    Note that Requirements~(1) and~(4) also ensure that exactly one proposition letter holds at each position.
     \item \emph{Correct ordering of cell indexes, row constraint, and well-formedness}: we require that for  each infix $\eta$ of the given word $w$   having length
  $2$, the following holds:
   \begin{itemize}
     \item \emph{either} the two positions of $\eta$ have the same tag and consecutive cell indexes $i$ and $i+1$ for some $1\leq i<n$ such that $[d_i]_\Right = [d_{i+1}]_\Left$, where $d_i$ (resp., $d_{i+1}$) is the domino-type of position with cell index $i$ (resp., $i+1$),
     \item \emph{or} the two positions of $\eta$ have distinct tag  and cell indexes $1$ and $n$, respectively.
   \end{itemize}
   This ensures, in particular, that the permutation $\pi$ in Requirement~(4) corresponds either to the natural ordering of cell indexes or to its reverse. Requirement~(4) can be easily captured by the  \hsDhom{} formula
\[
\begin{array}{ll}
   \varphi_{C} := & \bD\bigl(len_{=2} \longrightarrow \\
   & \bigl[ \displaystyle{\bigvee_{d,d'\in\Delta\mid [d]_\Right = [d']_\Left}\,\,\bigvee_{b\in\{0,1\}}\,\,\bigvee_{i\in [1,n-1]}}
   (\D(d,i,b)\wedge \D(d',i+1,b))\,\vee
    \vspace{0.1cm} \\
    & \displaystyle{\bigvee_{d,d'\in\Delta}\,\,\bigvee_{b\in\{0,1\}}}  (\D(d,1,b)\wedge \D(d',n,1-b))\,\bigr]
\end{array}
\]
\end{enumerate}

By construction, it easily follows that a word $w$ is a model of $\varphi_{wf}$ if and only if (i) $w$ visits some position with the initial domino-type and cell index $1$, and (ii) $w$ or its reverse corresponds to an infix $\eta$ of some well-formed sequence  of row-codes
such that $\eta$ contains some row-code.
This concludes the proof of the statement about satisfiability checking of Theorem~\ref{theorem:hardness} (assuming the  $\subint$-semantics).

\section{Conclusions}
In this paper, we have determined the exact complexity of  the satisfiability checking and model checking problem for the logic \hsD\ of sub-intervals over finite linear orders and finite Kripke structures, respectively, under the homogeneity assumption. 
Since the two problems have been proved to be $\Psp$-complete, we have enriched the set of known ``tractable'' fragments  of HS (see \cite{TCSBMMPS19}) with a meaningful logic. 

The tractability of \hsD\ gives an additional insight in the open problem of the existence of elementary decision procedures for model checking the logic \hsBE\ and the full logic HS\@. \hsD\ can be considered as the most meaningful logic included in \hsBE, and its tractability leaves open the possibility of a positive answer to the open 
question. 

A way to tackle the question is to investigate the HS fragments featuring the modality \hsD\ together with either the modality \hsB\ (for the relation \emph{started-by}) or the modality \hsE\ (for the relation \emph{finished-by}) so to gradually approach the expressive power of \hsBE. 

The complexity of model checking the logic \hsBE\ and the full HS logic remains the most intriguing and challenging open question.

\paragraph*{Acknowledgements.}
The work by Alberto Molinari, Angelo Montanari, and Pietro Sala has been supported by the GNCS project \emph{Logic and Automata for Interval Model Checking}.

\bibliographystyle{alphaurl}
\bibliography{biblio}

\newcommand{\etalchar}[1]{$^{#1}$}
\begin{thebibliography}{BMM{\etalchar{+}}19b}

\bibitem[BDG{\etalchar{+}}14]{DBLP:journals/amai/BresolinMGMS14}
D.~Bresolin, D.~{Della Monica}, V.~Goranko, A.~Montanari, and G.~Sciavicco.
\newblock The dark side of interval temporal logic: marking the undecidability
  border.
\newblock {\em Annals of Mathematics and Artificial Intelligence},
  71(1-3):41--83, 2014.
\newblock \href {https://doi.org/10.1007/s10472-013-9376-4}
  {\path{doi:10.1007/s10472-013-9376-4}}.

\bibitem[BGMS09]{DBLP:journals/apal/BresolinGMS09}
D.~Bresolin, V.~Goranko, A.~Montanari, and G.~Sciavicco.
\newblock Propositional interval neighborhood logics: Expressiveness,
  decidability, and undecidable extensions.
\newblock {\em Annals of Pure and Applied Logic}, 161(3):289--304, 2009.
\newblock \href {https://doi.org/10.1016/j.apal.2009.07.003}
  {\path{doi:10.1016/j.apal.2009.07.003}}.

\bibitem[BGMS10]{DBLP:journals/logcom/BresolinGMS10}
D.~Bresolin, V.~Goranko, A.~Montanari, and P.~Sala.
\newblock Tableaux for logics of subinterval structures over dense orderings.
\newblock {\em Journal of Logic and Computation}, 20(1):133--166, 2010.
\newblock \href {https://doi.org/10.1093/logcom/exn063}
  {\path{doi:10.1093/logcom/exn063}}.

\bibitem[BMM{\etalchar{+}}17]{bmmpsICALP}
L.~Bozzelli, A.~Molinari, A.~Montanari, A.~Peron, and P.~Sala.
\newblock {Satisfiability and Model Checking for the Logic of Sub-Intervals
  under the Homogeneity Assumption}.
\newblock In {\em Proceedings of the 44th International Colloquium on Automata,
  Languages, and Programming (ICALP)}, pages 120:1--120:14. Schloss
  Dagstuhl--Leibniz-Zentrum fuer Informatik, 2017.
\newblock \href {https://doi.org/10.4230/LIPIcs.ICALP.2017.120}
  {\path{doi:10.4230/LIPIcs.ICALP.2017.120}}.

\bibitem[BMM{\etalchar{+}}18]{ICBMMPS18}
L.~Bozzelli, A.~Molinari, A.~Montanari, A.~Peron, and P.~Sala.
\newblock Model checking for fragments of the interval temporal logic {HS} at
  the low levels of the polynomial time hierarchy.
\newblock {\em Information and Compututation}, 262(Part):241--264, 2018.
\newblock \href {https://doi.org/10.1016/j.ic.2018.09.006}
  {\path{doi:10.1016/j.ic.2018.09.006}}.

\bibitem[BMM{\etalchar{+}}19a]{TOCLMMPS19}
L.~Bozzelli, A.~Molinari, A.~Montanari, A.~Peron, and P.~Sala.
\newblock Interval vs. point temporal logic model checking: An expressiveness
  comparison.
\newblock {\em {ACM} Trans. Comput. Log.}, 20(1):4:1--4:31, 2019.
\newblock \href {https://doi.org/10.1145/3281028} {\path{doi:10.1145/3281028}}.

\bibitem[BMM{\etalchar{+}}19b]{TCSBMMPS19}
L.~Bozzelli, A.~Molinari, A.~Montanari, A.~Peron, and P.~Sala.
\newblock Which fragments of the interval temporal logic {HS} are tractable in
  model checking?
\newblock {\em Theoretical Computer Science}, 764:125--144, 2019.
\newblock \href {https://doi.org/10.1016/j.tcs.2018.04.011}
  {\path{doi:10.1016/j.tcs.2018.04.011}}.

\bibitem[DGMS11]{DBLP:journals/eatcs/MonicaGMS11}
D.~{Della Monica}, V.~Goranko, A.~Montanari, and G.~Sciavicco.
\newblock Interval temporal logics: a journey.
\newblock {\em Bulletin of the {EATCS}}, 105:73--99, 2011.
\newblock URL:
  \url{http://albcom.lsi.upc.edu/ojs/index.php/beatcs/article/view/ 98}.

\bibitem[Har92]{harel92}
D.~Harel.
\newblock {\em Algorithmics: The spirit of computing}.
\newblock Wesley, 2nd edition, 1992.

\bibitem[HS91]{interval_modal_logic}
J.~Y. Halpern and Y.~Shoham.
\newblock A propositional modal logic of time intervals.
\newblock {\em Journal of the ACM}, 38:279--292, 1991.
\newblock \href {https://doi.org/10.1145/115234.115351}
  {\path{doi:10.1145/115234.115351}}.

\bibitem[KR93]{from_discourse_to_logic}
H.~Kamp and U.~Reyle.
\newblock {\em From Discourse to Logic: Introduction to Model-theoretic
  Semantics of Natural Language, Formal Logic and Discourse Representation
  Theory, Volume 42 of Studies in Linguistics and Philosophy}.
\newblock Springer, 1993.
\newblock \href {https://doi.org/10.1007/978-94-017-1616-1}
  {\path{doi:10.1007/978-94-017-1616-1}}.

\bibitem[LM13]{LM13}
A.~Lomuscio and J.~Michaliszyn.
\newblock An epistemic {H}alpern-{S}hoham logic.
\newblock In {\em Proceedings of the 23rd International Joint Conference on
  Artificial Intelligence (IJCAI)}, pages 1010--1016, 2013.

\bibitem[LM14]{LM14}
A.~Lomuscio and J.~Michaliszyn.
\newblock Decidability of model checking multi-agent systems against a class of
  {EHS} specifications.
\newblock In {\em Proceedings of the 21st European Conference on Artificial
  Intelligence (ECAI)}, pages 543--548, 2014.
\newblock \href {https://doi.org/10.3233/978-1-61499-419-0-543}
  {\path{doi:10.3233/978-1-61499-419-0-543}}.

\bibitem[LM16]{LM16}
A.~Lomuscio and J.~Michaliszyn.
\newblock Model checking multi-agent systems against epistemic {HS}
  specifications with regular expressions.
\newblock In {\em Proceedings of the 15th International Conference on
  Principles of Knowledge Representation and Reasoning (KR)}, pages 298--308,
  2016.

\bibitem[Lod00]{DBLP:conf/asian/Lodaya00}
K.~Lodaya.
\newblock Sharpening the undecidability of interval temporal logic.
\newblock In {\em Proceedings of the 6th Asian Computing Science Conference
  (ASIAN)}, pages 290--298, 2000.
\newblock \href {https://doi.org/10.1007/3-540-44464-5_21}
  {\path{doi:10.1007/3-540-44464-5_21}}.

\bibitem[MM14]{DBLP:journals/fuin/MarcinkowskiM14}
J.~Marcinkowski and J.~Michaliszyn.
\newblock The undecidability of the logic of subintervals.
\newblock {\em Fundamenta Informaticae}, 131(2):217--240, 2014.
\newblock \href {https://doi.org/10.3233/FI-2014-1011}
  {\path{doi:10.3233/FI-2014-1011}}.

\bibitem[MMM{\etalchar{+}}16]{DBLP:journals/acta/MolinariMMPP16}
A.~Molinari, A.~Montanari, A.~Murano, G.~Perelli, and A.~Peron.
\newblock Checking interval properties of computations.
\newblock {\em Acta Informatica}, 53(6-8):587--619, 2016.
\newblock \href {https://doi.org/10.1007/s00236-015-0250-1}
  {\path{doi:10.1007/s00236-015-0250-1}}.

\bibitem[MMP18]{ICMMP18}
A.~Molinari, A.~Montanari, and A.~Peron.
\newblock Model checking for fragments of halpern and shoham's interval
  temporal logic based on track representatives.
\newblock {\em Information and Computation}, 259(3):412--443, 2018.
\newblock \href {https://doi.org/10.1016/j.ic.2017.08.011}
  {\path{doi:10.1016/j.ic.2017.08.011}}.

\bibitem[Mon16]{DBLP:conf/time/Montanari16}
A.~Montanari.
\newblock Interval temporal logics model checking.
\newblock In {\em Proceedings of the 23rd International Symposium on Temporal
  Representation and Reasoning, ({TIME})}, page~2, 2016.
\newblock \href {https://doi.org/10.1109/TIME.2016.32}
  {\path{doi:10.1109/TIME.2016.32}}.

\bibitem[MPS10]{DBLP:conf/time/MontanariPS10}
A.~Montanari, I.~Pratt{-}Hartmann, and P.~Sala.
\newblock Decidability of the logics of the reflexive sub-interval and
  super-interval relations over finite linear orders.
\newblock In {\em Proceedings of the 17th International Symposium on Temporal
  Representation and Reasoning (TIME)}, pages 27--34, 2010.
\newblock \href {https://doi.org/10.1109/TIME.2010.18}
  {\path{doi:10.1109/TIME.2010.18}}.

\bibitem[Ott01]{DBLP:journals/jsyml/Otto01}
M.~Otto.
\newblock Two variable first-order logic over ordered domains.
\newblock {\em Journal of Symbolic Logic}, 66(2):685--702, 2001.
\newblock \href {https://doi.org/10.2307/2695037} {\path{doi:10.2307/2695037}}.

\bibitem[Pra05]{DBLP:journals/ai/Pratt-Hartmann05}
Ian Pratt{-}Hartmann.
\newblock Temporal prepositions and their logic.
\newblock {\em Artificial Intelligence}, 166(1-2):1--36, 2005.
\newblock \href {https://doi.org/10.1016/j.artint.2005.04.003}
  {\path{doi:10.1016/j.artint.2005.04.003}}.

\bibitem[Sha04]{DBLP:conf/aiml/Shapirovsky04}
I.~Shapirovsky.
\newblock On {PSPACE}-decidability in transitive modal logic.
\newblock In {\em Proceedings of the 5th conference on Advances in Modal logic
  (AiML)}, pages 269--287, 2004.
\newblock URL: \url{http://www.aiml.net/volumes/volume5/Shapirovsky.ps}.

\bibitem[SS05]{DBLP:journals/logcom/ShapirovskyS05}
I.~Shapirovsky and V.~B. Shehtman.
\newblock {Modal Logics of Regions and Minkowski Spacetime}.
\newblock {\em Journal of Logic and Computation}, 15(4):559--574, 2005.
\newblock \href {https://doi.org/10.1093/logcom/exi039}
  {\path{doi:10.1093/logcom/exi039}}.

\bibitem[vB91]{DBLP:journals/jphil/Benthem91}
J.~van Benthem.
\newblock Language in action.
\newblock {\em Journal of Philosophical Logic}, 20(3):225--263, 1991.
\newblock \href {https://doi.org/10.1007/BF00250539}
  {\path{doi:10.1007/BF00250539}}.

\bibitem[Ven90]{venema1990}
Y.~Venema.
\newblock Expressiveness and completeness of an interval tense logic.
\newblock {\em Notre Dame Journal of Formal Logic}, 31(4):529--547, 1990.
\newblock \href {https://doi.org/10.1305/ndjfl/1093635589}
  {\path{doi:10.1305/ndjfl/1093635589}}.

\bibitem[Ven91]{chopping_intervals}
Y.~Venema.
\newblock A modal logic for chopping intervals.
\newblock {\em Journal of Logic and Computation}, 1(4):453--476, 1991.
\newblock \href {https://doi.org/10.1093/logcom/1.4.453}
  {\path{doi:10.1093/logcom/1.4.453}}.

\end{thebibliography}
\end{document}